\RequirePackage{fix-cm}
\documentclass[smallextended]{svjour3}       
\smartqed  

\usepackage[utf8]{inputenc}

\usepackage{graphicx}
\usepackage{hw}
\usepackage[noend]{algpseudocode}
\usepackage{xcolor,multirow}
\usepackage{bm}
\usepackage{mathrsfs,tabularx}
\usepackage{multicol,enumitem,caption}

\usepackage{subcaption}

\def\comp{\ensuremath\mathop{\scalebox{.7}{$\circ$}}}
\newcommand*\dotv{{}\cdot{}}

\newcolumntype{L}[1]{>{\raggedright\arraybackslash}p{#1}}
\newcolumntype{C}[1]{>{\centering\arraybackslash}p{#1}}
\newcolumntype{R}[1]{>{\raggedleft\arraybackslash}p{#1}}

\begin{document}

\title{Nonconvex and Nonsmooth Sparse Optimization via Adaptively Iterative Reweighted Methods
}


\author{Hao Wang\textsuperscript{1} \and Fan Zhang\textsuperscript{1,3,4}  \and Yuanming Shi\textsuperscript{1} \and Yaohua Hu\textsuperscript{2}
\thanks{CONTACT: Hao Wang. Email: wanghao1@shanghaitech.edu.cn}
}


\institute{\textsuperscript{1}
              School of Information Science and Technology,
              ShanghaiTech University, Shanghai 201210, People's Republic of China           
           \and\\
           \textsuperscript{2}
           College of Mathematics and Statistics, Shenzhen
           Key Laboratory of Advanced Machine Learning and Applications, Shenzhen University, Shenzhen 518061, People's Republic of China           
           \and\\
           \textsuperscript{3}
           University of Chinese Academy of Sciences, Beijing 100049, People's Republic of China       
           \and\\
           \textsuperscript{4}
           Shanghai Institute of Microsystem and Information Technology, Chinese
           Academy of Sciences, Shanghai 200050, People's Republic of China 
}

\date{Received: date / Accepted: date}

\maketitle

\begin{abstract}
We propose a general formulation of nonconvex and nonsmooth sparse optimization problems with convex set constraint, which can take into account most existing types of nonconvex sparsity-inducing terms, bringing strong applicability to a wide range of applications. We design a general algorithmic framework of iteratively reweighted algorithms for solving the proposed nonconvex and nonsmooth sparse optimization problems, which solves a sequence of weighted convex { regularization} problems with adaptively updated weights.  First-order optimality condition is derived and global convergence results are provided under loose assumptions, making our theoretical results a practical tool for analyzing a family of various reweighted algorithms. The effectiveness and efficiency of our proposed formulation and the algorithms are demonstrated in numerical experiments on various sparse optimization problems. 
\end{abstract}

\section{Introduction}
\label{intro}
Nonconvex and nonsmooth sparse optimization problems
have been becoming a prevalent research topic in many disciplines of applied mathematics and engineering.  Indeed, there has been a tremendous increase in a number of application areas in which 
nonconvex sparsity-inducing techniques  have been employed, such as machine learning~\cite{bradley1998feature,ml2}, telecommunications~\cite{Shi2016Smoothed,gsbf}, image reconstruction~\cite{image1}, sparse recovery~\cite{daubechies2010iteratively,malek2014iterative} and signal processing~\cite{candes2008enhancing,miosso2009compressive}. 
This is mainly because of their superior ability to  reduce the complexity of a system, improve the generalization of the prediction performance, and/or enhance the robustness of the solution, compared with traditional convex sparsity-inducing techniques. 

Despite their wide application, nonconvex and nonsmooth sparse optimization problems are computationally challenging to solve
due to the nonconvex and nonsmooth nature of the sparsity-inducing terms. A popular method for handling the convex/nonconvex { regularization} problems is the iteratively reweighted algorithm, which approximates the nonconvex and nonsmooth problem by a sequence of trackable convex subproblems.  
There have been some  iteratively reweighted algorithms   proposed for special cases of the nonconvex and nonsmooth problems.  
{For example, in~\cite{chartrand2008iteratively,ling2013decentralized,lai2013improved} W. Yin \emph{et al}. have designed an  iteratively reweighted algorithm  for solving the unconstrained nonconvex $\ell_p$ norm model} and   in~\cite{lu2014iterative} Z. Lu has  analyzed the global convergence of a class of reweighted algorithms  for the unconstrained nonconvex $\ell_p$ regularized problem.  
The constrained $\ell_p$-regularization problem is studied to improve the image restoration using a priori information and  the optimality condition of this problem is given in~\cite{wang2015optimality}.  
The critical technique of this type of algorithms is to add relaxation parameters to transform the nonconvex and nonsmooth sparsity-inducing terms into smooth approximate functions and then use linearization to obtain convex subproblems~\cite{chen2010convergence,lai2011unconstrained}. It should be noticed that the relaxation parameter should be driven to 0 in order to obtain the solution of the original unrelaxed problem. 
Two most popular variants of iteratively reweighted type of algorithms are the iteratively reweighted $\ell_1$ minimization and the iteratively reweighted { squared $\ell_2$-norm minimization}. The former has  convex but nonsmooth subproblem, while  the latter leads to  convex and smooth  subproblems. 
It has been reported by E. Candes \emph{et al.} in~\cite{candes2008enhancing} that the reweighted $\ell_1$ minimization can significantly enhance sparsity of the solution.  It has been demonstrated that Iteratively reweighted least-squares have greatly promoted the computation and correctness of robust regression estimation~\cite{holland1977robust,street1988note}.

However, iteratively reweighted algorithms are generally difficult to track and analyze. This is mainly because most nonconvex functions are  non-Lipschitz continuous, especially around sparse solutions that we are particularly interested in.  The major issue caused by this situation is  that the optimal solution cannot be  characterized by common optimality conditions. For some special cases where the sparsity-inducing term is $\ell_p$-norm and no constraint is involved, the first-order and second-order sufficient optimality conditions have been studied  in~\cite{wang2015optimality,ochs2015iteratively}.  
X. Chen \emph{et al} have derived a first-order necessary optimality condition for local minimizers  and define the generalized stationary point of the constrained optimization
problems with nonconvex regularization~\cite{bian2017optimality}.
These results are used by Z. Lu to derive  the global convergence of a class of iteratively reweighted $\ell_1$ and $\ell_2$ methods for unconstrained $\ell_p$ regularization problems~\cite{lu2014iterative}. 
For the sum of a convex function and a (nonconvex) nondecreasing function applied to another convex function, the convergence to a critical point of the iteratively reweighted $\ell_1$ algorithm   is provided when the objective function must be coercive~\cite{ochs2015iteratively}. As far as we know, the convergence results of iteratively reweighted $\ell_1$ and $\ell_2$ methods for   Non-Lipschitz nonconvex and nonsmooth problems with general convex-set constraint have not been provided.

As for more general cases, the analysis in current work has many limitations due to  this obstacle for theoretical  analysis.  
First, instead of driving the relaxation parameter to 0,  many existing  methods~\cite{chartrand2008iteratively,EPS}  aim to show the convergence
to the optimal solution of the relaxed sparse optimization
problems. An iteratively reweighted least squares   algorithm for the relaxed problem $\ell_p$ problem  in sparse signal recovery has been discussed in~\cite{ba2014convergence}  with  local convergence rate analysis.
A critical aspect of any implementation of such an approach is the selection of the relaxation parameters which prevents the weights from becoming overwhelmingly large. As been explained in~\cite{chartrand2008iteratively}, large relaxation parameters will smooth out many local minimizers, whereas small values can cause  the subproblems difficult to solve and the algorithm too quickly get trapped into local minimizers. For this purpose, updating strategies of the relaxation parameter have been studied in~\cite{EPS}, but it is only designed  for constrained convex problems. 

Second,  some methods assume Lipschitz continuity of the objective  in their analysis, which only holds true for few nonconvex sparsity-inducing terms such as log-sum { regularization}~\cite{lu2014iterative,ochs2015iteratively}.  Though this assumption is not explicitly required by some other researchers,  they need another assumption that  the negative of the sparsity-inducing term--which is convex--is subdifferentiable everywhere~\cite{ochs2015iteratively}. However, it must be noticed that this is a quite strong assumption and  generally not suitable for most sparsity-inducing terms e.g. $\ell_p$-norm, consequently limiting their applicability to various cases. 

This situation may become even worse when a general convex set constraint is added to the problem.  To the best of our knowledge, only simple cases such as linearly constrained cases have been studies by current work. 
To circumvent the obstacle for analysis, current methods either focus on the relaxed problem as explained above, or unconstrained reformulations   where the constraint violation is  penalized in the objective~\cite{chartrand2008iteratively,lu2014iterative,wang2015optimality}.  The latter approach 
then arises the issue of  how to select the proper { regularization} parameter value. 

Moreover,   some work~\cite{lu2014proximal} assumes the coercivity  of the objective to guarantee that the iterates generated by the algorithms must have clustering points. 
This sets another limitation  as many sparsity-inducing terms is bounded above, e.g.,  arctan function. This assumption therefore requires the rest part of the objective must be coercive, which is generally not the case. 

Overall, for  cases involving  more general sparsity-inducing terms and convex-set constraints, the analysis of the behavior of iteratively reweighted algorithms remains an open question. 

{
Despite of the iteratively reweighted algorithms, the Difference of Convex (DC) algorithm is also used for tackling some specific-form sparse optimization problems. As J. Pang \emph{et al} mentioned in \cite{ahn2017difference,cao2018unifying}, a great deal of existing nonconvex regularizations can be represented as DC functions, and then can be solved by the DC algorithm. Unfortunately, these works need the assumption that the nonconvex { regularization} term is Lipschitz differentiable, which limits the applicability to the cases like $\ell_p$-norm { regularization}. 
To address this issue, T. Liu \emph{et al} have proposed a Successive Difference of Convex Approximation Method (SDCAM) in \cite{liu2019successive}, which makes use of the Moreau envelope as a smoothing technique for the nonsmooth terms. However, they need the proximal mapping of each nonsmooth function is easy to compute. This requirement obstacles the use of lots of nonconvex regularizations, of which the proximal operator can not be solved effectively. For example, when we use the $\ell_p$ norm ($0<p<1$) as the regularization term, only the cases $p=1/2$ and $p=2/3$ have an analytical solution of the associated proximal mapping as shown in \cite{cao2013fast}.
This fact limits the SDCAM to select other values of $p$. And as shown in experimental results, the $\ell_p$ regularization with different values of $p$ perform differently. Usually, smaller $p$ induces more sparse solution.

}

In this paper, we consider a unified formulation of the convex set constrained nonconvex and nonsmooth sparse optimization problems. A general algorithmic framework of Adaptively Iterative Reweighted (AIR) algorithm is presented for solving these problems. We   derive  the first-order  condition to characterize the optimal solutions and analyze global  convergence of the proposed method to the first-order optimal solutions. The most related research work mainly includes 
the iteratively  reweighted algorithms proposed by~\cite{EPS}  for solving general constrained convex problems,   the reweighted methods by~\cite{lu2014iterative,lu2017ell} for solving unconstrained $\ell_p$ regularization problems, and the algorithmic framework proposed  in~\cite{ochs2015iteratively}  for solving the unconstrained nonsmooth and nonconvex optimization problems. However, we emphasize again  that our  focus is on dealing with cases with general nonconvex and nonsmooth sparsity-inducing terms and general convex set constraints---a stark contrast to the situations considered by most existing methods. 

The contributions of this paper can be summarized as follows. 

\begin{itemize}
	
	\item Our unified problem formulation can take into account  most existing types of  nonconvex sparsity-inducing functions which also allows for   group structure  as well as    general convex set constraint. A general algorithmic framework of adaptively iteratively reweighted algorithms is developed    by  solving a sequence of trackable convex subproblems.  A unified first-order  necessary conditions is derived to characterize the optimal solutions  by using Fr\'echet subdifferentials, and global convergence of the proposed algorithm is provided. 
	
	\item For reweighted $\ell_1$ and $\ell_2$ minimizations, our algorithm allows for vanishing relaxation parameters, which can avoid the issue of selecting appropriate value of relaxation parameter. The  global convergence analysis for both iteratively reweighted $\ell_1$ and $\ell_2$ algorithms are provided.  We show that every limit point generated by the algorithms must satisfy the first-order necessary optimality condition for the original unrelaxed problem, instead of the relaxed problem---a novel result that most current work does not possess. E. Candes, M. Wakin and S. Boyd have put forward several open questions in~\cite{candes2008enhancing} about the reweighted $\ell_1$ algorithm, including under what conditions the reweighted $\ell_1$ algorithm will converge,    when the iterates have limit point, and the ways of updating $\epsilon$ as the algorithm progresses towards a solution. Our work  provides answers to these questions.

	\item 
	We have proven that the existence of the cluster point to the algorithm can be guaranteed by understanding conditions under which the iterates generated is bounded. The conditions guaranteeing the boundedness is also given.
	Conditions for selecting the starting point and the initial relaxation parameters are also  provided  to guarantee the global convergence. This makes our methods also apply for cases  where the objective is not coercive.
\end{itemize}

\subsection{Organization}
In the remainder of this section, we outline our notation
and introduce various concepts that will be employed throughout the paper. In \S~\ref{s.Pd},
we describe our problem of interests and explain its connection to various existing types of sparsity-inducing techniques.
In \S~\ref{s.Ad}, we describe the details of our proposed AIR method and apply it to different types of nonconvex sparsity-inducing terms. 
The optimality condition and the global convergence of the proposed algorithm in different situations are  provided in \S~\ref{s.Conv}. 
We discuss implementations of our methods and the
results of numerical experiments in in \S~\ref{sec.experiment}.  Concluding remarks are provided in \S~\ref{sec.conclusion}.

\subsection{Notation and preliminaries}
Much of the notation that we use is standard, and when it is not, a definition is provided. For convenience, we review some of this notation and preliminaries here. 

We use $\mathbf{0}$ to represent the vector filled with all 0s  of appropriate dimension. Let $\mathbb{R}^n$ be the space of real $n$-vectors, $\mathbb{R}^n_+$ be the nonnegative orthant of $\mathbb{R}^n$, $\mathbb{R}^n_+=
\{\xbf \in\mathbb{R}^n : \xbf \ge \mathbf{0}\}$  and the nonpositive orthant  
$\{\xbf \in\mathbb{R}^n : \xbf \le \mathbf{0}\}$. Moreover,  let  $\mathbb{R}^n_{++}$ be its interior $\mathbb{R}_{++}^n:=\{\xbf\in\mathbb{R}^n: 
\xbf >\mathbf{0} \}$.  The set of $m\times n$ real matrices is denoted by $\mathbb{R}^{m\times n }$.  
For a pair of vectors $(\ubf, \vbf)\in\mathbb{R}^n\times \mathbb{R}^n$, their inner product is written as 
$\la \ubf, \vbf\ra$.  The set of nonnegative integers is denoted by $\Nb$. 
Suppose  $\mathbb{R}^n$ be the product space of subspaces $\mathbb{R}^{n_i}, i=1,\ldots, m$ with $\sum_{i=1}^mn_i=n$, i.e., it takes decomposition $\mathbb{R}^n = \mathbb{R}^{n_1}\times \ldots \times \mathbb{R}^{n_m}$.
Given a closed convex set $X \subset \mathbb{R}^n$, the normal cone to $X$ at a point $\bar\xbf \in X$ is given by 
\[ N(\bar\xbf|X) := \{\zbf |  \langle \zbf, \xbf-\bar\xbf \rangle \leq 0,  \  \forall \xbf\in X\}.\]
The characteristic  function of $X$ is defined as 
\[ \delta(\xbf|X) = \begin{cases}  0 & \text{ if } \xbf\in X,\\
+ \infty & \text{ otherwise.}
\end{cases}
\] 
The indicator operator $\mathbb{I}(\cdot)$  is an indicator function that takes a value of 1 if the statement is true and 0 otherwise.

For a given $\alpha \in\mathbb{R}$,  denote the  level set of $f $ as  
\[ L (\alpha; f) := \{\xbf \in \mathbb{R}^n  |  f(\xbf)\le  \alpha\}.\]
In particular, we are interested in level set with an upper bound reachable for $f$: 
\[ L (f(\hat\xbf); f) := \{\xbf \in \mathbb{R}^n  |  f(\xbf)\le  f(\hat\xbf)\}.\]
The subdifferential of a convex function $f$ at $\xbf$ is a set defined by 
\[ \partial f(\xbf) = \{ \zbf \in \mathbb{R}^n  |  f(\ybf) - f(\xbf) \ge \la \zbf, \ybf - \xbf\ra, \forall \ybf\in\mathbb{R}^n\}. \]
Every  element $\zbf \in \partial f(\xbf)$ is referred to as a subgradient.   
To characterize the optimality conditions for nonsmooth problems, we need to 
introduce the concepts  of Fr\'echet subdifferentiation. 
In fact, there are a variety of   subdifferentials known by now including limiting subdifferentials, approximate subdifferentials 
and Clarke's generalized gradient, many of which 
can be used here for deriving the optimality conditions.  
The major tool we choose in this paper is the Fr\'echet subdifferentials, 
which were introduced in~\cite{bazaraa1974cones,kruger1981varepsilon} and discussed in~\cite{kruger2003frechet}. 

\begin{definition}[Fr\'echet subdifferential] Let $f$ be a function from a real Banach space into an extended real line 
	$\bar{\mathbb{R} }= \mathbb{R} \cup \{+\infty\}$, finite at $\xbf$.  The Fr\'echet subdifferential of $f$ at $\xbf$, denoted as 
	$\partial_F f(\xbf)$, is the set 
	$$
	\partial_F f(\xbf) = \left\{ \xbf^* \in \mathbb{R}^n:  \liminf\limits_{\ubf\to \xbf} \frac{f(\ubf)-f(\xbf)-\la \xbf^*, \ubf-\xbf\ra }{\|\ubf-\xbf\|} \ge 0\right\}.
	$$
	Its elements are referred to as Fr\'echet subgradients.
\end{definition}

For a composite function $r\comp c(\xbf)$, where $c: \mathbb{R}^n\to\mathbb{R}$ and $r: \mathbb{R} \to \mathbb{R}$,  denote 
$\partial_F r(c(\xbf))$ (or simply $\partial_F r(c)$) as the Fr\'echet subdifferential of $r$ with respect to $c$, and 
$r'(c(\xbf))$ (or simply $r'(c)$) as the derivative of $r$ with respect to $c(\xbf)$ if $r$ is differentiable at $c(\xbf)$. 

\section{Problem Statement and its applications}\label{s.Pd}
In this section, we propose a unified formulation of the constrained nonconvex and nonsmooth sparse optimization problem, and list the instances in some prominent applications.
\subsection{Problem Statements}
We consider the following nonconvex and nonsmooth sparse optimization problem
\begin{equation}\label{f.sparse}
\begin{aligned}
\min_{\xbf\in\mathbb{R}^n} & \quad  f(\xbf)+  \Phi(\xbf)\\
\st &\quad   \xbf \in  X,
\end{aligned}
\end{equation}
where $f:\mathbb{R}^n\to \mathbb{R}$ is smooth and convex   
and $X\subset\mathbb{R}^n$ is a closed convex set.  
This type of problem is a staple for many applications in signal processing~\cite{candes2008introduction,liu2018sparse}, wireless communications~\cite{qin2018sparse,shi2018generalized} and machine learning~\cite{bach2012optimization,hastie2015statistical}. For example, in signal processing, $f$ may be the mean-squared error for signal recovery, $X$ may be a nonnegative constraint for signal~\cite{khajehnejad2011sparse}; in wireless communications, $f$ may represent the system performance such as transmit power consumption, $X$ may be the transmit power constraints and quality of service constraints~\cite{gsbf};
in machine learning, $f$ can represent the convex loss function, such as the cross-entropy loss for logistic regression~\cite{harrell2015ordinal}.

In a large amount of applications, the being recovered vector $\xbf$ is expected to have some sparse property in a structured manner. To handle this type of structured sparsity, various types of group-based $\Phi$ has been studied in~\cite{wainwright2014structured}.
Consider a collection of groups $\Gcal = \{\Gcal_1,\Gcal_2,\cdots,\Gcal_m\}$ with $|\Gcal_i| = n_i$.
The union over all groups covers the full index set and $\sum_{i=1}^{m} n_i = n$. The structured vector $\xbf$ can be written as
\[\xbf = [ \underbrace{x_1,x_2,\cdots,x_{n_1}}_{\xbf_{\Gcal_1}^T},\cdots,\underbrace{x_{n-n_{m} + 1},\cdots,x_{n}}_{\xbf_{\Gcal_m}^T}]^T.
\]
With these ingredients, the associated group-based function $\Phi$ takes the form
\[ \Phi(\xbf) =   \sum\limits_{i=1}^m r_{i}(c_{i}(\xbf_{\Gcal_i})),\]
where  $c_i: \mathbb{R}^{n_i}\to \mathbb{R}$ is convex and $r_i: \mathbb{R}\to \mathbb{R}$ is concave for each $i$.
Throughout this paper, we make the following assumptions about $f$, $r_i$, $c_i$  and $X$. 

\begin{assumption} \label{basic assumption}
	The functions $f$, $r_i$, $c_i$, $i=1,\ldots, m$, and set $ X$   are such that 
	\begin{enumerate}
		\item[(i)] $ X$ is closed and convex. 
		\item[(ii)] $f$ is smooth, convex and bounded below by $\underline f$ on $ X$.
		\item[(iii)] $r_i$ is  smooth on $\mathbb{R}\setminus\{0\}$,   concave  and strictly increasing on $\mathbb{R}_{+}$ with $r_i(-c) = r_i(c)$ and $r_i(0) = 0$, and is  Fr\'echet subdifferentiable at 0.
		\item[(iv)] $c_i$ is convex and coercive with $c_i(\xbf_i)\ge0, \forall \xbf\in X$ where the equality holds if and only if $\xbf_i=\mathbf{0}$.
		\item[(i)]   The composite function $\phi_i=r_i\comp c_i(x) = r_i(c_i(\xbf_i))$ is concave on regions $\{ x \mid c_i(\xbf_i) \ge 0\}$ and 
		$\{ x \mid c_i(\xbf_i) \le 0\}$.
		
	\end{enumerate}
\end{assumption}

\begin{remark}
	The  symmetry of $r_i$ is not a requirement, since $c_i(\xbf) \ge 0$ is assumed always true; the purpose of this assumption is to simplify the analysis.
\end{remark}

Most existing sparse optimization problems can be reverted to~\eqref{f.sparse}.  In next subsection,   
we describe the important applications of problem~\eqref{f.sparse} 
and explain the specific forms of the functions  $f$, $r_i$, $c_i$ in the example. 
Based on different formulations of the composite function $\Phi(\xbf)$, there are a great deal of 
nonconex sparsity-inducing techniques to promote sparse solutions, such as the approximations of the $\ell_0$ norm of $\xbf$.

{The problems considered here allow the sparse-inducing terms  to be non-Lipschitz, and it also allows the presence of a 
	general convex set constraint. If some components of the stationary point tend to zero, it is possible that both the Fr\'echet subdifferentials of the non-Lipschitz term and the normal direction of the constraint tend to infinity. This leads to the difficulty of analyzing the convergence of their linear combination (i.e., the optimality condition residual of the problem).}      

\subsection{Sparsity-inducing Functions}

Many applications including signal processing, wireless communications and machine learning involve the minimization of  the $\ell_0$-norm of the variables $\|\xbf\|_0$, i.e., the number of nonzero components in $\xbf$. 
However,  this   is    regarded   as an NP-hard problem,  thus various approximations of $\ell_0$ norm have been proposed.  
By   different choice of the formulation  $r_i$ and $c_i$, there exist many approximations to $\ell_0$ norm, so that 
a smooth approximate problem of  \eqref{f.sparse} is derived with 
\[ \|\xbf\|_0\approx \Phi(\xbf) = \sum_{i=1}^n r_i(c_i(x_i)).\]
In the following discussion, we only provide the expression of $r_i$ on $\mathbb{R}_+$, since 
by Assumption~\ref{basic assumption}, $r_i$ can be defined accordingly on $\mathbb{R}_-$. 

\subsubsection{The EXP Approximation}
The first instance is the feature selection algorithm via concave minimization proposed by Bradley and Mangasarian~\cite{bradley1998feature} with approximation 
\begin{equation}\label{f.exp}\tag{EXP}
\|\xbf\|_0 \approx \sum_{i=1}^{n}1-e^{-p|x_i|}\  \text{ with }\  p> 0,
\end{equation}
where $p$ is chosen to be sufficiently large to promote sparse solutions. The concavity of this function leads to a finitely terminating algorithm and a more accurate
representation of the feature selection algorithm. It is reported  that the algorithms with this formulation obtained a reduction in error with selected features fewer in number and they are faster compared to traditional convex feature selection algorithms. 
For example, we can choose 
\[c_i(x_i) = |x_i|,\ r_i =1-e^{-pc_i}  \text{ or  }  c_i = x_i^2,\ r_i =1-e^{-p \sqrt{c_i}},\] 
so that this approximation  can be viewed  as a specific formulation of $\Phi$.

\subsubsection{The LPN Approximation}
The second instance, which is widely used in many applications currently, is to approximate 
the $\ell_0$ norm by $\ell_p$ quasi-norm~\cite{fazel2003log} 
\begin{equation}\label{f.lp}\tag{LPN}
\|\xbf\|_0 \approx \sum_{i=1}^{n}|x_i|^p \ \text{ with }  p\in(0,1)
\end{equation}
and  $p$ is chosen close to 0 to enforce sparsity in the solutions.  Based on this approximation, numerous applications and algorithms have emerged. 
Here we can choose 
\[ c_i(x_i) =|x_i|, r_i(c_i)=c_i^p  \text{ or  }  c_i(x_i)=x_i^2,   r_i (c_i)= c_i^{p/2}\]
in the formulation of $\Phi$.

\subsubsection{The LOG Approximation}
Another option for approximating $\ell_0$ norm, proposed in~\cite{lobo2007portfolio}, is to use the log-sum approximation  
\begin{equation}\label{f.log}\tag{LOG}
\|\xbf\|_0 \approx \sum_{i=1}^{n}\log{(1+p|x_i|)}  \  \text{ with } p > 0,
\end{equation}
and setting  $p$   sufficiently large leads to sparse solutions.  
We can choose  
$$ c_i(x_i) = |x_i|, r_i (c_i)= \log{(1+pc_i)},
$$
or  
$$  c_i(x_i)=x_i^2, r_i(c_i)=\log(1+p\sqrt{c_i}).$$

\subsubsection{The FRA Approximation}
The approximation technique proposed in~\cite{fazel2003log} suggests   
\begin{equation}\label{f.inv}\tag{FRA}
\|\xbf\|_0 \approx \sum_{i=1}^{n}\frac{|x_i|}{|x_i|+p}, \text{ with } p > 0,
\end{equation}
and   $p$ is required to be sufficiently small to promote sparsity.  One can use 
$$c_i(x_i) = |x_i|,\ r_i (c_i)= \frac{c_i}{c_i+p}, $$ or 
$$ c_i x_i= x_i^2,\ r_i (c_i)= \frac{\sqrt{c_i}}{\sqrt{c_i}+p}.$$

\subsubsection{The TAN Approximation}
Cand\`es et al.   propose an approximation to the $\ell_0$ norm in~\cite{candes2008enhancing}
\begin{equation}\label{f.atan}\tag{TAN}
\|\xbf\|_0 \approx \sum_{i=1}^{n}\arctan(p|x_i|), \text{ with } p>0, 
\end{equation}
and sufficiently small $p$ can cause sparsity in the solution. The function $\arctan$ is bounded above and $\ell_0$-like. 
It is reported that this approximation tends to work well and often  better than the log-sum \eqref{f.log}. In this case, we can choose  
$$c_i(x_i) = |x_i|,   r_i(c_i) = \arctan{(pc_i)}, $$ or  
$$ c_i(x_i)=x_i^2, r_i(c_i)=\arctan{(p\sqrt{c_i})}.$$

\subsubsection{The SCAD and MCP Approximation}
Another nonconvex sparsity-inducing technique needs to be mentioned is the SCAD { regularization} proposed in~\cite{fan2001variable}, which require the derivative of $\phi_i$ to satisfy 
\begin{equation}\label{f.scad}\tag{SCAD}
c_i(x_i)=|x_i|,  \phi_i'(c_i) = \lambda\{\mathbb{I}(c_i\le \lambda ) + \frac{(a\lambda - c_i)_+}{(a-1)\lambda} \mathbb{I}(c_i>\lambda)\},\end{equation}
for some $a > 2$, where often $a=3.7$ is used.  Alternatively, the MCP~\cite{zhang2007penalized} { regularization} uses 
\begin{equation}\label{f.mcp}\tag{MCP}
c_i(x_i)=|x_i|,  \phi_i'(c_i) = (a\lambda-c_i)_+/a   \text{ for some } a \ge 1.
\end{equation}

\subsubsection{The Group Structure}
	These sparsity-inducing functions can also take into account group structures. For example, $\ell_{p,q}$-norm with $p\ge 1$ and $0<q<1$~\cite{hu2017group} is defined as
	$$\|\xbf\|_{p,q} = \(\sum_{i=1}^m \|\xbf_{\Gcal_i}\|_p^q\)^{1/q}. $$
	Therefore,   we can choose
	$$\Phi(\xbf) = \|\xbf\|_{p,q}^q,\ \text{ with }  c_i(\xbf_{\Gcal_i}) = \|\xbf_{\Gcal_i}\|_p\ \text{ and }\  r_i(c_i)=c_i^q.$$

\subsection{Problem Analysis}
There have been various literatures for solving the nonconvex and nonsmooth sparse optimization problems. In~\cite{chartrand2008iteratively,ling2013decentralized} Wotao Yin \emph{et al}. have considered solve the sparse signal recovery problem by using the unconstrained nonconvex $\ell_p$ norm model, proposed the associated iterative reweighted unconstrained $\ell_p$ algorithm and provided the convergence analysis for the reweighted $\ell_2$ case. In~\cite{lu2014iterative} Zhaosong Lu have provied the first-order optimality condition for the unconstrained nonconvex $\ell_p$ norm problem, and convergence analysis for both $\ell_1$ and $\ell_2$ types reweighted algorithm. However, it is not clear for analyzing the first-order optimalizty condition for the constrained nonconvex and nonsmooth sparse optimization problem~\eqref{f.sparse}. In order to address this issue, we propose the AIR algorithm in \S~\ref{s.Ad}, provide the first-order optimality condition for~\eqref{f.sparse} and the convergence analysis for the AIR algorithm in \S~\ref{s.Conv}.

\section{Adaptively Iterative Reweighted Algorithm}\label{s.Ad}
In this section, we present the adaptively iterative reweighted  algorithm  for minimizing the nonconvex and nonsmooth sparse optimization problem~\eqref{f.sparse}. 
\subsection{Smoothing Method}
In this subsection, we show how we deal with the nonsmoothness.
Before proceeding, we define the following functions for  $\xbf\in X$. 
Problem \eqref{f.sparse}  can be rewritten as 
\beq\label{prob.J0} \min_{\xbf}\ J_0(\xbf) := f(\xbf)+  \sum_{i \in \Gcal}r_i(c_i(\xbf_i )) + \delta(\xbf|X).\eeq
Adding  relaxation parameter $\epsbf \in \mathbb{R}_+^m$ to smooth the (possibly) nondifferentiable $r_i$, we propose the relaxed problem as
\beq\label{prob.J}
\min_{\xbf}\ J(\xbf; \epsbf):= f(\xbf)+  \sum_{i\in\Gcal}r_i(c_i(\xbf _i)+\eps_i)  + \delta(\xbf|X),
\eeq
and in particular, $J(\xbf; \mathbf{0}) = J_0(\xbf)$. Here we extend the notation of $\phi_i$ and use  $\phi_i(\xbf_i; \eps_i)$ to denote the relaxed sparsity-inducing function, so that 
$$
\phi_i(\xbf_i; \eps_i) := r_i(c_i(\xbf_i)+\eps_i),
$$
$$
\Phi(\xbf; \epsbf): = \sum_{i\in\Gcal}\phi_i(\xbf_i;\eps_i) \  \text{ and }\  \phi_i(\xbf_i)=\phi_i(\xbf_i; \mathbf{0}).
$$
The following theorem shows that the pointwise convergence of $J(\xbf;\epsbf)$ to $J_0(\xbf)$ on $X$ as $\epsbf\to\mathbf{0}$.  
\begin{theorem}\label{thm.j approx j0}
	For any $\xbf\in X$ and $\eps\in\mathbb{R}_{++}$, it holds true that 
	\[J_0(\xbf)  \le   J(\xbf; \epsbf)  \le  J_0(\xbf) 
	+ \sum_{c_i(\xbf_i) = 0} r_i(\eps_i) +\sum_{ c_i(\xbf_i) > 0} r'(c_i(\xbf_i)) \eps_i.
\]
	This implies that $J(\xbf;\epsbf)$ pointwise convergence to $J_0(\xbf)$ on $X$ as $\epsbf\to\mathbf{0}$.
\end{theorem}
\begin{proof}
	The first inequality is trivial, so we only have to show the second inequality.  
	Since $r(\dotv )$ is concave on $\mathbb{R}_+$, we have
	\beq\label{eq.lem.2.1}
	r_i(z) \le r_i(z_0) + r_i'(z_0) (z-z_0)\quad \text{ for any } z, z_0 \in\mathbb{R}_+,\eeq
	Therefore,  
	\[\begin{aligned}
	J(\xbf; \epsbf)  = & f(\xbf) + \sum_{i\in\Gcal} r_i(c_i(\xbf_i)+\eps_i) \\ 
	= & f(\xbf) +  \sum_{c_i(\xbf_i) = 0} r_i(\eps_i)  +    \sum_{ c_i(\xbf_i) > 0} r_i(c_i(\xbf_i)+\eps_i) \\
	\le &  f(\xbf)   + \sum_{c_i(\xbf_i) = 0} r_i(\eps_i) + \sum_{c_i(\xbf_i) > 0} r_i(c_i(\xbf_i)) +  \sum_{c_i(\xbf_i) > 0}   r'(c_i(\xbf_i)) \eps_i \\
	= & J_0(\xbf)  + \sum_{c_i(\xbf_i) = 0} r_i(\eps_i) +\sum_{ c_i(\xbf_i) > 0} r'(c_i(\xbf_i)) \eps_i,
	\end{aligned}
	\]
	where the inequality follows by~\eqref{eq.lem.2.1}. This completes the first statement.
	
	On the other hand, since 
	\[ \lim\limits_{\bm{\eps}\to 0}\sum_{c_i(\xbf_i) = 0} r_i(\eps_i) +\sum_{ c_i(\xbf_i) > 0} r'(c_i(\xbf_i)) \eps_i = 0,
	\] 
	it holds
	\[\lim\limits_{\epsbf\to \mathbf{0}} J(\xbf; \epsbf) = J_0(\xbf),\quad \xbf\in X.
	\]
\end{proof}

\subsection{Adaptively Iterative Reweighted Algorithm}

A convex and smooth function $G_{(\tilde{\xbf}, \tilde \epsbf)}(\xbf)$ can be derived as an approximation of $J(\tilde \xbf,\tilde \epsbf)$ at $\tilde\xbf$  
by linearizing $r_i$ at $c_i(\tilde\xbf_i)+\tilde \eps_i$ to have the subproblem
\beq\label{prob.G}
G_{(\tilde{\xbf}, \tilde \epsbf)}(\xbf) :=  f(\xbf)+   \sum_{i\in\Gcal} w_i(\tilde \xbf_i, \tilde \eps_i) c_i(\xbf_i )  + \delta(\xbf|X) ,
\eeq
where the weights are given by  
\[ w_i(\xbf,\eps_i)=  r_i'(c_i(\xbf_i)+\eps_i),\quad i\in\Gcal.\]
Note that the relaxation parameter can be simply chosen as $\epsbf = \mathbf{0}$ if $r$ is smooth at 0.

At iterate $\xbf^k$, the new iterate   is   obtained by
\[ \xbf^{k+1} \in \arg\min_{\xbf}   G_{(\xbf^k, \epsbf^k)}(\xbf).\]
Therefore,  $\xbf^{k+1} $ satisfies optimality condition
\[ 0\in \partial  G_{(\xbf^k,\epsbf^k)}(\xbf^{k+1} ) .\]
The relaxation parameter is selected such that  $\epsbf^{k+1} \le \epsbf^k$ and possibly  driven to 0 as the algorithm proceeds.

Our proposed \emph{A}daptively \emph{I}terative \emph{R}eweighted algorithm for nonconvex and nonsmooth sparse optimization problems is stated in Algorithm~\ref{alg.1}.

\begin{algorithm}
	\caption{AIR: Adaptively Iterative Reweighted}
	\label{alg.1}
	\begin{algorithmic}[1]
		\State (Initialization)  Choose $\xbf^0\in X$ and  $\epsbf^0\in\mathbb{R}_{++}^n$. Set $k=0$.
		\State (Subproblem Solution)  Compute new iterate
		\[ \xbf^{k+1}\in\arg\min   G_{(\xbf^k, \epsbf^k)}(\xbf).\]
		\State (Reweighting)
		Choose $\epsbf^{k+1} \in(\mathbf{0},  \epsbf^k]$.
		\State Set $k\gets k+1$. Go to Step 2.
	\end{algorithmic}
\end{algorithm}

\subsection{$\ell_1$-Algorithm \& $\ell_2$-Algorithm}\label{sub.L12}
In this subsection, we describe the details of how to construct $G_{(\tilde \xbf,\tilde\epsbf)}(\xbf)$ for the nonconvex and nonsmooth sparsity-inducing functions 
\eqref{f.exp}--\eqref{f.mcp}  in \S~\ref{s.Pd}.
Notice that the relaxation parameter  $\eps_i$ could set as $0$ if $\lim\limits_{c_i\to0+}r_i'(c_i)<+\infty$.   
For simplicity, denote 
$ \tilde w_i = w_i(\tilde \xbf_i, \tilde\eps_i).$  In Table~\ref{tab.1},  we provide the explicit forms of the weights $\tilde w_i$ at $(\tilde \xbf_i,\tilde \epsilon_i)$ when choosing  $c_i(\xbf_i) = \| \xbf_i \|_1$ and   $c_i(\xbf_i) = \|\xbf_i\|_2^2$  for each case, so that the corresponding subproblem is an $\ell_1$-norm sparsity-inducing problem and an $\ell_2$-norm sparsity-inducing problem 
\[\begin{aligned}
G_{(\tilde \xbf,\tilde\epsbf)}(\xbf) & = f(\xbf) + \sum\limits_{i\in\Gcal} \tilde w_i \|\xbf_i\|_1+ \delta(\xbf|X)\quad\text{and }\\
G_{(\tilde \xbf,\tilde\epsbf)}(\xbf)  & = f(\xbf) + \sum\limits_{i\in\Gcal} \tilde w_i \|\xbf_i\|_2^2 +\delta(\xbf|X).
\end{aligned}
\]
For each sparsity-inducing function, we consider $c_i(\xbf_i)= \|\xbf_i \|_1$ in the first row and $c_i(\xbf_i)=\|\xbf_i\|_2^2$ in the second row.  We also list the properties of the 
$r_i$ with $c_i\to\infty$ and its side-derivative of $r_i$ at $0$ in the fourth and fifth columns. This is because these properties can lead to different behaviors of each AIR as shown in the theoretical analysis.  

\begin{table}[h]
	\caption{Different AIR weights based on different choice of $r_i$ and $c_i$.}
	\label{tab.1}
	\centering
	\begin{tabular}{|C{3em}|C{5.5em}|C{13em}|C{3em}|C{3em}|}\hline
		$\phi_i $       &   $r_i(c_i)$   &   $\tilde w_i$ & $r_i(\infty)$ & $r_i'(0+)$  \\ \hline
		\multirow{2}{*}{\eqref{f.exp}}      &  $1-e^{-pc_i}$  & $pe^{-p(\|\tilde \xbf_i\|_1)}$  & $<\infty$ & $<\infty$ \\  \cline{2-5}
		&  $1-e^{-p\sqrt{c_i}}$ & $\frac{pe^{-p\tilde \xbf_i}}{2\tilde \xbf_i}$ &  $<\infty$ & $<\infty$ \\ \hline      
		\multirow{2}{*}{\eqref{f.lp}}       &  $c_i^p$  & $p(\|\tilde \xbf_i\|_1+\tilde\epsilon_i)^{p-1}$  &   $+\infty$ & $+\infty$ \\  \cline{2-5}
		&  $c_i^{p/2}$ & $\frac{p}{2}(\|\tilde \xbf_i\|_2^2+\epsilon_i)^{\frac{p}{2}-1}$ &    $+\infty$& $+\infty$ \\ \hline     
		\multirow{2}{*}{\eqref{f.log} }      &  $\log(1+pc_i)$      & $\frac{p}{1+p\|\tilde \xbf_i\|_1}$  &  $+\infty$ & $<\infty$ \\  \cline{2-5}
		&  $ \log(1+p\sqrt{c_i})$ & $\frac{p}{2\sqrt{\|\tilde \xbf_i\|_2^2+\tilde\epsilon_i}(1+p\sqrt{\|\tilde \xbf_i\|_2^2+\tilde\epsilon_i})}$ &  $+\infty$ & $+\infty$ \\ \hline  
		\multirow{2}{*}{\eqref{f.inv}}    &  $\frac{c_i}{c_i+p}$      & $\frac{p}{(\|\tilde \xbf_i\|_1+p)^2}$  &  $<\infty$ & $<\infty$ \\  \cline{2-5}
		&  $\frac{\sqrt{c_i}}{\sqrt{c_i}+p}$ & $\frac{p}{2\sqrt{\|\tilde \xbf_i\|_2^2+\epsilon_i}(\sqrt{\|\tilde \xbf_i\|_2^2+\epsilon_i}+p)^{2}}$ &  $<\infty$ & $ +\infty$ \\ \hline  
		\multirow{2}{*}{\eqref{f.atan}  }  &  $\frac{c_i}{c_i+p}$      & $\frac{p}{1+p^2(\|\tilde \xbf_i\|_1)^2}$  &  $<\infty$ & $<\infty$ \\  \cline{2-5}
		&  $c_i^{p/2}$ & $\frac{p}{2\sqrt{\|\tilde \xbf_i\|_2^2+\epsilon_i}(1+p^2(\|\tilde \xbf_i\|_2^2+\epsilon_i))}$ &  $<\infty$ & $ +\infty$ \\ \hline            
		
	\end{tabular}
\end{table} 

As for SCAD and MCP, the explicit forms of $r_i$ are not necessary to be known, but it can be easily verified using $r_i'$ that Assumption~\eqref{basic assumption} 
still holds true.  The  reweighted   $\ell_1$ subproblem for SCAD has weights  
\[ \tilde w_i =  \lambda\{\mathbb{I}(|\tilde x_i |+\tilde\eps_i\le \lambda )   + \frac{(a\lambda - \|\tilde \xbf_i\|_1 - \tilde\eps_i)_+}{(a-1)\lambda} \mathbb{I}(\|\tilde \xbf_i\|_1+ \tilde\eps_i>\lambda)\}.\]
The weights of reweighted $\ell_2$ subproblem for SCAD are   
\[\begin{aligned} \tilde w_i =  &  \frac{\lambda}{2\sqrt{\|\tilde \xbf_i\|_2^2+\epsilon_i)}}\{\mathbb{I}(\sqrt{\|\tilde \xbf_i\|_2^2+\epsilon_i)}\le\lambda) \\
& + \frac{(a\lambda-\sqrt{\|\tilde \xbf_i\|_2^2+\epsilon_i)})_{+}}{(a-1)\lambda}\mathbb{I}(\sqrt{\|\tilde \xbf_i\|_2^2+\epsilon_i)}>\lambda)\}.
\end{aligned}\]
As for MCP, the reweighted $\ell_1$ subproblem has weights 
\[\tilde w_i  =  (a\lambda-\|\tilde \xbf_i\|_1-\tilde\eps_i)_+/a,\]
and the weights for reweighted $\ell_2$ subproblem are 
\[\tilde w_i =  (a\lambda-\sqrt{\|\tilde \xbf_i\|_2^2+\epsilon_i})_+/a.\]

\section{Convergence analysis}
\label{s.Conv}

In this section, we analyze the global convergence   of our proposed AIR. First we provide a unified first-order optimality condition for the constrained nonconvex and nonsmooth sparse optimization problem~\eqref{f.sparse}. Then we establish the global convergence anlaysis followed by the existence of cluster points.

For simplicity, denote 
$w^k_i = w_i(\xbf^k, \eps_i^k),  \wbf_i^k = w_i^k \ebf_{n_i},$
$ \wbf^k = [ \wbf_1^k; \wbf_2^k; \ldots; \wbf_m^k],$ and  $\Wbf^k = \text{diag}(\wbf^k),$
and so forth.

\subsection{First-order Optimality Condition}
In this subsection, we    derive conditions to characterize the optimal solution of \eqref{f.sparse}. 
Due to the nonconvex and nonsmooth nature of the sparsity-inducing function, we use  Fr\'echet subdifferentials as the major tool in our analysis. 
Some important properties of Fr\'echet subdifferentials derived in~\cite{kruger2003frechet}  that will be used in this paper are  summarized below. Part $(i)$-$(iv)$ are Proposition 1.1, 1.2, 1.10, 1.13 and  1.18 in~\cite{kruger2003frechet}, respectively.
\begin{proposition}\label{prop.sub}
	The following statements about Fr\'echet subdifferentials is true. 
	\begin{enumerate}
		\item[(i)] If $f$ is 
		differentiable at $\xbf$ with gradient  $\nabla f(\xbf)$, then $\partial_F f(\xbf) = \{\nabla f(\xbf)\}$. 
		\item[(ii)]  If $f$ is convex, then $\partial_F f(\xbf) = \partial f(\xbf)$.
		\item[(iii)]  If $f$ is Fr\'echet subdifferential at $\xbf$ and attains local minimum at $\xbf$, then 
		\[ \mathbf{0}\in\partial_F f(\xbf).\]  
		\item[(iv)] Let $r(\cdot )$ be Fr\'echet subdifferentiable at $c^*=c(\xbf^*)$ with $c(\xbf)$ being convex, 
		then $r\comp c(\xbf)$ is Fr\'echet subdifferentiable at $\xbf^*$ and that 
		\[y^*\partial c(\xbf^*) \subset \partial_F r\comp c(\xbf^*)\]
		for any $y^*\in\partial_F r(c^*)$. 
		\item[(v)] $ N(\xbf |   X) = \partial_F\delta(\xbf |   X)$ if $X$ is   closed and convex. 
	\end{enumerate}
\end{proposition}

The properties of Fr\'echet subdifferentials in Proposition~\ref{prop.sub} can be used to characterize the optimal solution 
of \eqref{f.sparse}. The following theorem is straightforward from Proposition~\ref{prop.sub}, which describes  the necessary optimality condition of problem \eqref{f.sparse}.
\begin{theorem}\label{thm.opt cond}
	If   \eqref{prob.J} attains a local minimum at $\xbf$, then it holds true that 
	\begin{equation}\label{optimality}
	\mathbf{0} \in \partial_F J(\xbf; \epsbf) = \nabla f(\xbf) +  \partial_F\Phi(\xbf; \epsbf) + N(\xbf |  X).
	\end{equation}
	
\end{theorem}

Next we shall     further investigate the properties of $\partial_F \phi(\xbf; \epsbf)$. 
\begin{lemma}\label{lem.1aaa}
	It holds that 
	\[\nabla f(\xbf) +   \prod_{i\in\Gcal} y_i \partial c_i(\xbf_i)  +  N(\xbf|X)  \subset  \partial_F J(\xbf; \epsbf)\]  
	for any $y_i\in\partial_F r_i(c_i(\xbf_i)+\eps_i)$.
\end{lemma}
 \begin{proof}
 	
 	Note that $\phi(\xbf; \epsbf)$ takes structure 
 	\[\Phi(\xbf;\epsbf) = \sum_{i\in\Gcal} \phi_i(\xbf_i;\eps_i)\quad \text{ with } \phi_i(\xbf_i;\eps_i)= r_i( c_i(\xbf_i)+\eps_i).\]  Thus we can write the Fr\'echet subdifferentials of $\Phi$
 	$$
 	\begin{aligned}
 	\partial_F \Phi(\xbf;\epsbf) &=\prod_{i\in\Gcal}  \partial_F \phi_i(\xbf_i;\eps_i)   \\
 	&= \partial_F \phi_1(\xbf_1;\eps_1)\times\ldots\times \partial_F \phi_m(\xbf_m;\eps_m),
 	\end{aligned} 
 	$$
 	meaning that 
 	\[  \partial_F J(\xbf;\epsbf) =  \nabla f(\xbf) +   \prod_{i\in\Gcal}  \partial_F \phi_i(\xbf_i;\eps_i) 
 	+ N(\xbf|X).\]
 	On the other hand, every $c_i$ is assumed to be convex. From Proposition~\ref{prop.sub}, we know that 
 	\[ y_i  \partial c_i(\xbf_i) \subset   \partial_F \phi_i(\xbf_i;\eps_i),\quad \forall y_i\in\partial_F r_i(c_i(\xbf_i)+\eps_i),\]
 	completing the proof. 
 \end{proof}

 If  $c_i(\xbf_i) > 0$ or $\eps_i > 0$,    $r_i$ is   differentiable at  $c_i+\eps_i$
 so that $\partial_F \phi_i(\xbf_i;\eps_i) = r'_i(c_i(\xbf_i^*)+\eps_i) \partial  c_i(\xbf^*)$ by Proposition~\ref{prop.sub}.  
 Of particular interests are the properties of $\partial_F r_i(0)$.  
 Notice that $r'_i$ is decreasing on $\mathbb{R}_{++}$.  
 We investigate  $\partial_F \phi_i(\xbf_i;\eps_i) $ bases on  the limits (possibly infinite) in the lemma below. 
 
 \begin{lemma}\label{sub diff r}
 	Let $ y_i^*:= \lim\limits_{ c_i \to 0^+} r'_i(c_i)  \ge 0$. It holds true that 
 	\[
 	\begin{cases}
 	\partial_F r_i( c_i)  =  r'_i(c_i)    & \text{     if  }\    c_i > 0\\
 	\ \partial_F r_i (0) =  [-  y_i^*,   y_i^*],    &\text{     if  }\    y_i^* < +\infty,\\
 	\ \partial_F r_i (0) =  \mathbb{R},                       &\text{   if }\    y_i^* = +\infty,
 	\end{cases}
 	\]
 	so that 
 	\begin{enumerate}
 		\item If $c_i(\xbf^*_i) + \eps_i > 0,$
 		$$\partial_F \phi_i(\xbf_i;\eps_i) =     r'_i(c_i(\xbf_i^*)+\eps_i)\partial  c_i(\xbf^*_i);$$
 		\item If $c_i(\xbf^*_i)+\eps_i = 0,  y_i^* < +\infty,$
 		$$y_i \partial c_i(\xbf^*_i)   \subset  \partial_F \phi_i(\xbf_i;\eps_i), \ \forall y_i\in[-  y_i^*,   y_i^*];$$
 		\item If $c_i(\xbf^*_i)+\eps_i = 0, y_i^* = +\infty,$
 		$$y_i \partial c_i(\xbf^*_i)  \subset  \partial_F \phi_i(\xbf_i;\eps_i),\ \forall y_i\in\mathbb{R}.$$
 	\end{enumerate}
 \end{lemma}
 \begin{proof} The statement about the case that $c_i(\xbf^*_i) > 0$ is obviously true.  We only need consider the case that   $c_i(\xbf^*_i) = 0$.  
 	Notice that  
 	\[ \liminf_{c_i \to 0^+} \frac{r_i(c_i) - r_i(0)}{c_i} = \liminf_{0 < \tilde c_i < c_i \atop c_i \to 0^+}  r_i'(\tilde c_i)
 	= r'_i(0+)=   y_i^* \ge 0\]
 	by Assumption~\ref{basic assumption}$(iii)$. 
 	It can be easily verified by~\cite[Proposition  1.17] {kruger2003frechet} that 
 	\[
 	\partial_F r_i(0)  =  
 	\begin{cases}  [-  y_i^*,   y_i^*]  & \text{if  }    y_i^* < +\infty,\\
 	\mathbb{R} & \text{if }    y_i^* = +\infty.
 	\end{cases}
 	\]
 	It then follows from Proposition~\ref{prop.sub}$(iv)$  that 
 	\[\left\{
 	\begin{aligned} y_i \partial c_i(\xbf^*_i) & \subset  \partial_F \phi_i(\xbf_i;\eps_i), \forall y_i\in[-  y_i^*,   y_i^*],    &\text{if  }    y_i^* < +\infty,\\
 	y_i \partial c_i(\xbf^*_i)  & \subset \partial_F \phi_i(\xbf_i;\eps_i), \forall y_i\in\mathbb{R},                    &\text{if }    y_i^* = +\infty.
 	\end{aligned}\right.
 	\]
 	
 \end{proof}
 
 Note that we only require $\epsbf\in\mathbb{R}_+$.  If $\epsbf = \mathbf{0}$,  all the results we have derived for $J(\dotv; \epsbf)$ in this subsection also hold for $J_0$.

 \subsection{Global Convergence of The AIR Algorithm}
 
 In this subsection, we analyze the global convergence of AIR under Assumption~\ref{basic assumption}.  First of all, we need to show that the subproblem always has a solution.    For $\hat\epsbf\in\mathbb{R}_{++}$, the subproblem is obviously well-defined on $X$ since the weights $w_i^k = r'_i(\xbf^k_i+\eps_i^k) < +\infty$. To guarantee the proposed AIR is well defined, we must show the existence of the subproblem solution.  
 We have  the following lemma about the solvability of the subproblems. 
 
 \begin{lemma} For $\epsbf^k \in\mathbb{R}_{++}$,  $\arg\min_{\xbf} G_{(\xbf^k,\epsbf^k)}(\xbf)$ is    nonempty, so that $\xbf^{k+1}$ is well-defined. 
 \end{lemma}
 
 \begin{proof}  Pick $\tilde\xbf\in X$ and let $\alpha:= G_{(\xbf^k,\epsbf^k)}(\tilde\xbf)$.   The level set 
 	\[  \{ \xbf \in X | G_{(\xbf^k,\epsbf^k)} (\xbf) \le G_{(\xbf^k,\epsbf^k)}(\tilde\xbf) \}   \] 
 	must be nonempty since it contains $\tilde\xbf$,  and bounded  due to the coercivity of $w_i^k c_i$, $i\in\Gcal$ and the lower boundedness of $f$ on $X$. 
 	This completes the proof  by~\cite[Theorem 4.3.1]{ortega1970iterative}.
 \end{proof}

 We have the following key facts about solutions to \eqref{prob.G}, which implies that the new iterate $\xbf^{k+1}$ causes a decrease in the model   $J(\xbf, \epsbf^k)$.
 
 \begin{lemma}\label{lem.J decrease.case1}  
 	Let $\tilde\xbf \in X$, $\hat\epsbf, \tilde\epsbf\in\mathbb{R}_{++}^m$ with $\hat\epsbf\le \tilde\epsbf$ and $\tilde w_i = w_i(\tilde\xbf_i, \tilde\eps_i)$ for $i\in\Gcal$.
 	Suppose that $\hat\xbf \in \arg\min_{\xbf\in X} G_{(\tilde \xbf, \tilde \epsbf)}(\xbf)$.
 	Then, for any $k$, it holds true that 
 	\[ J(\hat\xbf,\hat\epsbf) - J(\tilde\xbf, \tilde\epsbf) \le G_{(\tilde\xbf, \tilde\epsbf)}(\hat\xbf ) - G_{(\tilde\xbf , \tilde\epsbf)}(\tilde\xbf) \le 0.\]
 \end{lemma}
 
 \begin{proof}
 	
 	First of all, 
 	$\hat \xbf  \in\arg\min_\xbf   G_{(\tilde\xbf, \tilde\epsbf)}(\xbf)$, so that  
 	$G_{(\tilde\xbf, \tilde\epsbf)}(\hat\xbf ) - G_{(\tilde\xbf , \tilde\epsbf)}(\tilde\xbf) \le 0$. 
 	Hence
 	$$\begin{aligned} J(\hat\xbf; \hat \epsbf ) \le & J(\hat \xbf; \tilde \epsbf) =  f(\hat \xbf ) +  \sum_{i\in\Gcal}
 	r_i( c_i(\hat \xbf_i ) + \tilde\eps_i )\\
 	\le & f(\tilde\xbf)+f(\hat\xbf )-f(\tilde\xbf) +  \sum_{i\in\Gcal}
 	r_i( c_i(\tilde\xbf) +\tilde\eps_i)+  \sum_{i\in\Gcal}
 	\tilde w_i ( c_i(\hat\xbf )-c_i(\tilde\xbf)  )\\
 	= & J(\tilde\xbf; \tilde\epsbf) + [G_{(\tilde\xbf, \tilde\epsbf)}(\hat\xbf )-G_{(\tilde\xbf,\tilde\epsbf)}(\tilde\xbf)],\\
 	\end{aligned}
 	$$
 	where the second inequality follows from \eqref{eq.lem.2.1}.
 \end{proof}
 
  Lemma~\ref{lem.J decrease.case1} indicates $J(\xbf; \epsbf)$ is monotonically decreasing for any $\xbf^0\in X, \epsbf^0\in \mathbb{R}^m$. 
  Define the model reduction 
  \[ \Delta G_{(\xbf^k,\epsbf^k)}(\xbf^{k+1}) = G_{(\xbf^k,\epsbf^k)}(\xbf^k) - G_{(\xbf^k,\epsbf^k)}(\xbf^{k+1}).\]
  The next lemma indicates this model reduction converges to zero, which naturally follows from Lemma~\ref{lem.J decrease.case1}. 
  
  \begin{lemma}\label{lem.delta to 0}
  	Suppose  
  	$\xbf^0\in  X$, $\epsbf^0 \in \Rb^m_{++}$, and $\{\xbf^k\}$ are generated by AIR.  The following statements hold true
  	\begin{enumerate}
  		\item[(i)] The sequence $\{\xbf^k\} \subset L(J(\xbf^0; \epsbf^0); J_0)$.
  		\item[(ii)]    $\lim\limits_{k\to\infty}\Delta G_{(\xbf^k,\epsbf^k)}(\xbf^{k+1})\to 0$.
  	\end{enumerate}
  \end{lemma}
   \begin{proof}
   	Part $(i)$ follows naturally  from   the fact that 
   	\[ J_0(\xbf^k) \le  J(\xbf^k,\epsbf^k) \le J(\xbf^0, \epsbf^0),\]
   	for all $k \in \Nb$ by Lemma~\ref{lem.J decrease.case1}. 
   	
   	For part $(ii)$, by Assumption~\ref{basic assumption},  $\tilde{J}:=\inf\limits_kJ(\xbf^{k}; \epsbf^k) > -\infty$. It follows from Lemma~\ref{lem.J decrease.case1}, that 
   	$$
   	J(\xbf^{k+1},\epsbf^{k+1}) \le J(\xbf^k,\epsbf^k) - \Delta G_{(\xbf^k,\epsbf^k)}(\xbf^{k+1} ).
   	$$
   	Summing up both sides of the above inequality from $0$ to $t$, we have
   	$$
   	\begin{aligned}
   	0&\le \sum\limits_{k=1}^{t} \Delta G_{(\xbf^k,\epsbf^k)}(\xbf^{k+1}) \\
   	&\le J(\xbf^{0},\epsbf^{0}) - J(\xbf^{t+1},\epsbf^{t+1})\le J(\xbf^{0},\epsbf^{0}) - \tilde{J}.
   	\end{aligned}
   	$$
   	Letting $t\to\infty$,  we know part $(ii)$ holds true.
   \end{proof}
   
   \subsubsection{Convergence Analysis for Bounded Weights}
   We first analyze the convergence when $\epsbf^k\to\epsbf^*\in\mathbb{R}_{++}$  or $\lim\limits_{c_i\to 0^+}r'_i(c_i) < + \infty$, $i\in\Gcal$. 
   In this case, $w_i^k \to w_i^*<+\infty$ if $\xbf_i^k\to \mathbf{0}$.
   The  ``limit subproblem'' takes form  
   \beq\label{general.sub.alter.0}
   \min_{\xbf}  \quad  \widetilde G_{(\tilde\xbf,\tilde\epsbf)}(\xbf) := f(\xbf) + \sum_{i\in\Gcal} \tilde w_i c_i(\xbf_i) + \delta(\xbf | X).  
   \eeq
   The existence of the solution to \eqref{general.sub.alter.0} is shown in the next lemma. 
   
   \begin{lemma}\label{lem.limit problem.0}
   	For $\tilde\epsbf\in\mathbb{R}_{++}$,  the optimal solution set of \eqref{general.sub.alter.0} is  nonempty.  Furthermore, if $\tilde\xbf$ is an optimal solution of \eqref{general.sub.alter.0}, then $\tilde\xbf$ also satisfies the first-order optimality condition of {\eqref{prob.J}}. 
   \end{lemma}
   \begin{proof}
   	Notice that 
   	$\tilde\xbf$ is feasible for \eqref{general.sub.alter.0} by the definition of $\widetilde G$. 
   	The level set 
   	\[  \{ \xbf \in X \mid  \widetilde G_{(\xbf^k,\epsbf^k)} (\xbf) \le \widetilde G_{(\xbf^k,\epsbf^k)}(\tilde\xbf)\}  \] 
   	must be nonempty since it contains $\tilde\xbf$ and  bounded  due to the coercivity of $\tilde w_i c_i$, $i\in\Gcal$ and the lower boundedness of $f$ on $X$. 
   	This completes the proof  by~\cite[Theorem 4.3.1]{ortega1970iterative}.   
   	
   	Therefore, any optimal solution $\xbf$ must satisfies  
   	\[ \mathbf{0}=  \nabla f(\xbf)_i +  \zbf_i + \nubf_i, i\in\Gcal \]
   	where $\nubf\in N(\xbf | X)$, $\zbf_i=\tilde w_i \xibf_i $ with 
   	\[ \tilde w_i \in\partial_F r_i(c_i(\tilde\xbf_i)+\tilde\epsbf_i),\ \xibf_i\in\partial c_i(\xbf_i),\ i\in\Gcal.\] 
   	The KKT conditions thus can be rewritten as following by Lemma~\ref{sub diff r}
   	$$\mathbf{0}=  \nabla f(\xbf)_i + \tilde w_i\xibf_i  + \nubf_i,$$  $$  \tilde w_i\in\partial_F r_i(c_i(\tilde \xbf_i) + \tilde \eps_i),\ \xibf_i\in\partial c_i(\xbf_i),$$
   	where $ \ i\in\Gcal$. 
   	If $\tilde\xbf$ is an optimal solution, then we have 
   	\[ \mathbf{0} \in \nabla f(\tilde\xbf) + \partial_F \Phi(\tilde\xbf;\tilde\epsbf) + N(\tilde\xbf | X),\]
   	implying $\tilde\xbf$ is optimal for $J(\dotv; \tilde\epsbf)$.
   \end{proof}
   
   Now we are ready to prove our main result in this section.
   
   \begin{theorem}\label{main.1.0}
   	Suppose     $\{\xbf^k\}_{k=0}^\infty$ is generated by AIR with initial point
   	$\xbf^0 \in X$ and  relaxation vector $\epsbf^0 \in \bR^m_{++}$ with $\epsbf^k\to\epsbf^*$.  Assume either  
   	\[ \eps_i^*>0 \text{ or }   r'(0+) < +\infty, \  i\in\Gcal\]
   	is true. Then  if $\{ \xbf^k\}$ has any cluster point,   it satisfies the optimality condition \eqref{optimality} for $J(\xbf; \epsbf^*)$. 
   \end{theorem}

   \begin{proof} Let $\xbf^*$ be a cluster point of $\{\xbf^k\}$.  From Lemma~\ref{lem.limit problem.0}, it suffices to show that 
   	$\xbf^*\in\arg\min_{\xbf}\widetilde G_{(\xbf^*,\epsbf^*)}(\xbf)$. 
   	We prove this by contradiction.  Assume that there exists a point $\bar\xbf$ such that 
   	$\varepsilon := \widetilde G_{(\xbf^*, \epsbf^*)}(\xbf^*) - \widetilde G_{(\xbf^*, \epsbf^*)}(\bar{\xbf}) >0$. 
   	Suppose  $\{\xbf^k\}_\Scal\to\xbf^*$,  $\Scal \subset \mathbb{N}$. Based on Lemma~\ref{lem.delta to 0}$(ii)$, there exists $k_1 >0$, such that for all $k> k_1$
   	\begin{equation}\label{g is small 0}
   	\widetilde G_{(\xbf^k, \epsbf^k)}(\xbf^k) - \widetilde G_{(\xbf^k, \epsbf^k)}(\xbf^{k+1})\le \varepsilon/4.
   	\end{equation}  
   	To derive a contradiction, notice that  $\xbf_i^k\overset{\Scal}{\to} \xbf_i^*$ and   $w_i^k \overset{\Scal}{\to} w_i^*$. There exists $k_2$ such that for all $k>k_2, k\in\Scal$,
   	\[
   	\begin{aligned}
   	\sum_{i\in \Gcal}(w_i^*-w_i^k)c_i(\bar\xbf_i) & \ge - \varepsilon/12,  \\
   	\sum_{i\in \Gcal}(w_i^k c_i(\xbf_i^k)-w_i^*c_i( \xbf_i^*) ) & \ge - \varepsilon/12,\\
   	f(\xbf^k) - f(\xbf^*) & \ge - \varepsilon/12.
   	\end{aligned}
   	\]
   	Therefore,  for all $k>k_2, k\in\Scal$,
   	\[
   	\begin{aligned}
   	&\ \widetilde G_{(\xbf^*, \epsbf^*)}(\xbf^*) - \widetilde G_{(\xbf^k, \epsbf^k)}(\bar\xbf)\\
   	= &\  [f(\xbf^*)+ \sum_{i\in \Gcal} w_i^* c_i(\xbf_i^*) ] - [f(\bar\xbf)+ \sum_{i\in \Gcal} [w_i^*-(w_i^*-w_i^k)] c_i(\bar\xbf_i) \\
   	= &\ [\widetilde G_{(\xbf^*, \epsbf^*)}(\xbf^*) - \widetilde G_{(\xbf^*, \epsbf^*)}(\bar\xbf) ] + \sum_{i\in \Gcal}(w_i^*-w_i^k)c_i(\bar\xbf_i),\\
   	\ge &\  [\widetilde G_{(\xbf^*, \epsbf^*)}(\xbf^*) - \widetilde G_{(\xbf^*, \epsbf^*)}(\bar\xbf) ] - \varepsilon/12\\
   	= &\ \varepsilon - \varepsilon/12 = 11\varepsilon/12,
   	\end{aligned}
   	\]
   	and that
   	\[
   	\begin{aligned}
   	& \ \widetilde G_{(\xbf^k, \epsbf^k)}(\xbf^k) - \widetilde G_{(\xbf^*, \epsbf^*)}(\xbf^*) \\
   	=  &\ [f(\xbf^k)   + \sum_{i\in \Gcal} w_i^k c_i(\xbf_i^k)]- [f(\xbf^*)+ \sum_{i\in \Gcal}w_i^* c_i(\xbf_i^*) ]\\
   	\ge  & \ -\varepsilon/6
   	\end{aligned}
   	\]
   	Hence, for all $k>\max(k_1,k_2), k\in\Scal$,  it holds that 
   	\[
   	\begin{aligned}
   	& \widetilde G_{(\xbf^k, \epsbf^k)}(\xbf^k) - \widetilde G_{(\xbf^k, \epsbf^k)}(\bar\xbf) \\
   	= & \widetilde G_{(\xbf^k, \epsbf^k)}(\xbf^k) - \widetilde G_{(\xbf^*, \epsbf^*)}(\xbf^*) + \widetilde G_{(\xbf^*, \epsbf^*)}(\xbf^*) - \widetilde G_{(\xbf^k, \epsbf^k)}(\bar\xbf)\\
   	= &11\varepsilon/12 -\varepsilon/6 = 3\varepsilon/4,
   	\end{aligned}
   	\]
   	contradicting with \eqref{g is small 0}. Therefore, $\xbf^*\in\arg\min_{\xbf}\widetilde G_{(\xbf^*,\mathbf{0})}(\xbf)$. 
   	By Lemma~\ref{lem.limit problem.0}, $\xbf^*$ satisfies the first-order optimality for \eqref{prob.J}.
   \end{proof}
   
   {\begin{remark}
   	Note that there are several choices of $\phi_i$ satisfy $r'(0+) < +\infty, \  i\in\Gcal$ as shown in Table \ref{tab.1}. For example, $\phi_i$ takes the \eqref{f.exp} form with both $c_i(\bm{x}_i) = \|\bm{x}_i\|_1$ and $c_i(\bm{x}_i) = \|\bm{x}_i\|_2^2$, $\phi_i$ takes the \eqref{f.log} form with $c_i(\bm{x}_i) = \|\bm{x}_i\|_1$ and $\phi_i$ takes the \eqref{f.atan} form with $c_i(\bm{x}_i) = \|\bm{x}_i\|_1$.
	\end{remark}}
   
   \begin{remark}
   	The convexity of $f$ is not necessary if $\xbf^{k+1}$ is found as the global minimizer of \eqref{prob.J}. In this case, the global convergence  we have derived so far can be modified accordingly, and in  the statement of Lemma~\ref{lem.limit problem.0},  a global minimizer $\tilde\xbf$ of 
   	\eqref{general.sub.alter.0} implies its optimality of \eqref{prob.J}.  
   \end{remark}
   
   \subsubsection{Convergence Analysis for Reweighed $\ell_1$ {and $\ell_2$} with Vanishing $\epsilon$ }

   We have shown the convergence of AIR with fixed $\epsbf$.  By Theorem~\ref{thm.j approx j0}, we can choose sufficiently small 
   $\epsbf$ and minimize $J(\cdot; \epsbf)$ instead of $J_0$ to obtain an approximate solution. However, as also shown  by Theorem~\ref{thm.j approx j0},  $J(\cdot; \epsbf)$   converges to $J_0$ only pointwisely. It then may be difficult to assert that the minimizer of $J(\cdot; \epsbf)$ is sufficiently 
   close to the minimizer of $J_0$ for given $\epsbf$.  Therefore, we consider to minimize  a sequence of $J(\cdot; \epsbf)$ with $\epsbf$ driven to $\mathbf{0}$.

   As the algorithm proceeds, of particular interest is the properties  of  the ``limit subproblem'' as the (sub)sequence of iterates converges. Let $\Lcal: = \{ i \mid r_i'(0+) = +\infty\}$. 
   Notice that it may happen
   $w_i^k\to\infty$ if $\xbf^k_i\to 0$ and $\eps_i^k \to 0$, so that $G$ may be not well-defined. Therefore we consider an alternative form of the ``limit 
   subproblem'' for $\tilde\epsbf=\mathbf{0}$
   \beq\label{general.sub.alter}
   \begin{aligned}
   	\min_{\xbf} & \quad  \widetilde G_{(\tilde\xbf,\mathbf{0})}(\xbf) := f(\xbf) + \sum_{i\in\Ncal(\tilde\xbf,\mathbf{0})} \tilde w_i c_i(\xbf_i) + \delta(\xbf | X),\\  \st&\quad  \xbf_i = \mathbf{0},\  i\in \Acal(\tilde\xbf,\mathbf{0}),
   \end{aligned}
   \eeq
   where $\Acal(\tilde\xbf,\mathbf{0}):=\{i\mid \tilde\xbf_i = \mathbf{0}, \tilde \eps_i= 0 \} \cap \Lcal$ and $\Ncal(\tilde\xbf,\mathbf{0}):=\Gcal\setminus\Acal(\tilde\xbf, \mathbf{0})$.   
   The existence of the solution to \eqref{general.sub.alter} is shown in the next lemma. 
   
   \begin{lemma}\label{lem.limit problem}
   	For $\tilde\epsbf=\mathbf{0}$,  the optimal solution set of \eqref{general.sub.alter} is  nonempty.  Furthermore, if $\tilde\xbf$ is an optimal solution of \eqref{general.sub.alter}, then $\tilde\xbf$ also satisfies the first-order optimality condition of \eqref{prob.J0}. 
   \end{lemma}
   \begin{proof}
   	Notice that 
   	$\tilde\xbf$ is feasible for \eqref{general.sub.alter} by the definition of $\widetilde G$. 
   	The level set 
   	\[  \{ \xbf \in X \mid  \widetilde G_{(\tilde\xbf,\mathbf{0})} (\xbf) \le \widetilde G_{(\tilde\xbf,\mathbf{0})}(\tilde\xbf); \  \xbf_i = \mathbf{0}, i\in\Acal(\tilde\xbf, \mathbf{0})\}  \] 
   	must be nonempty since it contains $\tilde\xbf$ and  bounded  due to the coercivity of $\tilde w_i c_i$, $i\in\Gcal$ and the lower boundedness of $f$ on $X$. 
   	This completes the proof  by~\cite[Theorem 4.3.1]{ortega1970iterative}.

   	Obviously  Slater's condition holds at any feasible point of \eqref{general.sub.alter}.  Therefore, any optimal solution $\xbf$ must satisfies the KKT conditions 
   	\[ \mathbf{0}=  \nabla f(\xbf)_i +  \zbf_i + \nubf_i, i\in\Gcal \]
   	with $\nubf\in N(\xbf | X)$, $\zbf_i=\tilde y_i \xibf_i $.
   	{
   	For the case $c_i(\xbf_i) = \|\xbf_i\|_1$, let $\tilde y_i:= \tilde w_i= r_i'(c_i(\tilde\xbf_i)), \xibf_i\in\partial c_i(\xbf_i), i\in\Ncal(\tilde\xbf,\mathbf{0})$.
   	Now for $i\in \Acal(\tilde\xbf,\mathbf{0})$, let $\tilde y_i = \|\zbf_i\|_\infty$ and $\xibf_i=  \zbf_i/\|\zbf_i\|_\infty$ so that $\xibf_i \in \partial c_i(0) = \partial c_i(\tilde\xbf_i)$. 
   	For the case $c_i(\xbf_i) = \|\xbf_i\|_2^2$, let $\tilde y_i:= \tilde w_i= r_i'(c_i(\tilde\xbf_i)), \xibf_i = 2\tilde \xbf_i, i\in\Ncal(\tilde\xbf,\mathbf{0})$.
   	Now for $i\in \Acal(\tilde\xbf,\mathbf{0})$, let $\tilde y_i = \|\zbf_i\|_\infty$ and $\xibf_i=  \zbf_i/\|\zbf_i\|_\infty$ so that $\xibf_i \in \partial \|\mathbf{0}\|_1$. The KKT conditions can be rewritten as 
   	$$ \mathbf{0}=  \nabla f(\xbf)_i +  \tilde y_i \xibf_i  + \nubf_i,$$
   	$$
   	\text{when}\ i \in \Ncal(\tilde\xbf,\mathbf{0}): \left\{\begin{matrix}
   	y_i\in\partial_F r_i(c_i(\tilde \xbf_i)),\\ 
   	\xibf_i\in\partial c_i(\xbf_i)
   	\end{matrix}\right.
   	$$
   	$$
   	\text{when}\ i \in \Acal(\tilde\xbf,\mathbf{0}): \left\{\begin{matrix}
   	y_i\in\partial_F r_i(\sqrt{c_i(\tilde \xbf_i)}),\\ 
   	\xibf_i\in\partial \sqrt{c_i(\xbf_i)}
   	\end{matrix}\right.
   	$$
   	
   	For both cases, the KKT conditions can be rewritten as 
   	$$ \mathbf{0}=  \nabla f(\xbf)_i +  \tilde y_i \xibf_i  + \nubf_i,$$    $$y_i\in\partial_F r_i(c_i(\tilde \xbf_i)),$$  $$\xibf_i\in\partial c_i(\xbf_i),  i\in\Gcal $$
   	by Lemma~\ref{sub diff r}. 
   	If $\tilde\xbf$ is an optimal solution, then we have 
   	\[ \mathbf{0} \in f(\tilde\xbf) + \partial_F \phi(\tilde\xbf;\mathbf{0}) + N(\tilde\xbf | X),\]
   	implying $\tilde\xbf$ is optimal for $J_0(\dotv)$.
   }
   \end{proof}
   
   Now we are ready to prove our main result in this section.
   
   \begin{theorem}\label{main.1}
   	Suppose   sequence $\{\xbf^k\}_{k=0}^\infty$ is generated by AIR $\ell_1$-algorithm with initial point
   	$\xbf^0 \in X$ and relaxation vector $\epsbf^0 \in \bR^m_{++}$.
   	If  $\{ \xbf^k\}$ has any cluster point $\xbf^*$, then it satisfies the optimality condition for $J_0(\xbf)$.
   \end{theorem}
   
   \begin{proof} Let $\xbf^*$ be a cluster point of $\{\xbf^k\}$ and $\lim\limits_{k\to\infty} \epsbf^k=\mathbf{0}$. There exists $k_0>0$ such that for all $k>k_0$, $\epsbf^k=\mathbf{0}$.
   	From Lemma~\ref{lem.limit problem}, it suffices to show that 
   	$\xbf^*\in\arg\min_{\xbf}\widetilde G_{(\xbf^*,\mathbf{0})}(\xbf)$. 
   	We prove this by contradiction.  Assume that there exists a point $\bar\xbf$ such that $c_i(\bar\xbf_i)= 0$ for all $i \in\mathcal{A}(\xbf^*,\mathbf{0})$ and $G_{(\xbf^*, \mathbf{0})}(\xbf^*) - G_{(\xbf^*, \mathbf{0})}(\bar{\xbf})>\varepsilon>0$. 
   	Suppose  $\{\xbf^k\}_\Scal$,  $\Scal \subset \mathbb{N}$. Based on Lemma~\ref{lem.delta to 0}$(ii)$, there exists $k_1 >0$, such that for all $k> k_1$
   	\begin{equation}\label{g is small}
   	\widetilde  G_{(\xbf^k, \epsbf^k)}(\xbf^k) - \widetilde G_{(\xbf^k, \epsbf^k)}(\xbf^{k+1})\le \varepsilon/4.
   	\end{equation}  
   	Notice that  $\xbf_i^k\overset{\Scal}{\to} \xbf_i^*$ and   $w_i^k \overset{\Scal}{\to} w_i^*$.
   	To derive a contradiction, there exists $k_2>k_0$ such that for all $k>k_2, k\in\Scal$,
   	\[
   	\begin{aligned}
   	\sum_{i\in \Ncal(\xbf^*, \mathbf{0})}(w_i^*-w_i^k)c_i(\bar\xbf_i) & > - \varepsilon/12,  \\
   	\sum_{i\in \Ncal(\xbf^*, \mathbf{0})}(w_i^k c_i(\xbf_i^k)-w_i^*c_i( \xbf_i^*) ) & > - \varepsilon/12,\\
   	f(\xbf^k) - f(\xbf^*) & > - \varepsilon/12.
   	\end{aligned}
   	\]
   	Therefore,  for all $k>k_2, k\in\Scal$,
   	\[
   	\begin{aligned}
   	&\ \widetilde G_{(\xbf^*, \mathbf{0})}(\xbf^*) -  \widetilde G_{(\xbf^k, \mathbf{0})}(\bar\xbf)\\
   	= &\  [f(\xbf^*)+ \sum_{i\in \Ncal(\xbf^*, \mathbf{0})} w_i^* c_i(\xbf_i^*) ] -[f(\bar\xbf)+ \sum_{i\in \Ncal(\xbf^*, \mathbf{0})} [w_i^*-(w_i^*-w_i^k)] c_i(\bar\xbf_i)] \\
   	= &\ [\widetilde G_{(\xbf^*, \mathbf{0})}(\xbf^*) - \widetilde G_{(\xbf^*, \mathbf{0})}(\bar\xbf) ] + \sum_{i\in \Ncal(\xbf^*, \mathbf{0})}(w_i^*-w_i^k)c_i(\bar\xbf_i),\\
   	\ge &\  [\widetilde G_{(\xbf^*, \mathbf{0})}(\xbf^*) - \widetilde G_{(\xbf^*, \mathbf{0})}(\bar\xbf) ] - \varepsilon/12\\
   	\ge &\ \varepsilon - \varepsilon/12 = 11\varepsilon/12,
   	\end{aligned}
   	\]
   	and that
   	\[
   	\begin{aligned}
   	& \ \widetilde G_{(\xbf^k, \mathbf{0})}(\xbf^k) -  \widetilde G_{(\xbf^*, \mathbf{0})}(\xbf^*) \\
   	=  &\ [f(\xbf^k) + \sum_{i\in \Acal(\xbf^*, \mathbf{0})} w_i^k c_i(\xbf_i^k) + \sum_{i\in \Ncal(\xbf^*, \mathbf{0})} w_i^k c_i(\xbf_i^k)]\\
   	&- [f(\xbf^*)+ \sum_{i\in \Ncal(\xbf^*, \mathbf{0})}w_i^* c_i(\xbf_i^*) ]\\
   	\ge &\ [f(\xbf^k)   + \sum_{i\in \Ncal(\xbf^*, \mathbf{0})} w_i^k c_i(\xbf_i^k)]- [f(\xbf^*)+ \sum_{i\in \Ncal(\xbf^*, \mathbf{0})}w_i^* c_i(\xbf_i^*) ]\\ 
   	\ge  & \ -\varepsilon/6
   	\end{aligned}
   	\]
   	Hence, for all $k>\max(k_1,k_2), k\in\Scal$,  it holds that 
   	\[
   	\begin{aligned}
   	& \  \widetilde G_{(\xbf^k, \mathbf{0})}(\xbf^k) - \widetilde G_{(\xbf^k, \mathbf{0})}(\xbf^{k+1}) \\
   	= & \  \widetilde G_{(\xbf^k, \mathbf{0})}(\xbf^k) - \widetilde G_{(\xbf^*, \mathbf{0})}(\xbf^*) + \widetilde G_{(\xbf^*, \mathbf{0})}(\xbf^*) - \widetilde G_{(\xbf^k, \mathbf{0})}(\bar\xbf)\\
   	= &11\varepsilon/12 -\varepsilon/6 = 3\varepsilon/4,
   	\end{aligned}
   	\]
   	contradicting with \eqref{g is small}. Therefore, $\xbf^*\in\arg\min_{\xbf}\widetilde G_{(\xbf^*,\mathbf{0})}(\xbf)$. 
   	By Lemma~\ref{lem.limit problem}, $\xbf^*$ satisfies the first-order optimality for \eqref{prob.J}.
   \end{proof}

   \subsection{Existence of Cluster Points}
   
   We will show that our proposed algorithm AIR is a descent method for the function $J(\xbf, \epsbf)$.  Consequently, both the existence of solutions to 
   \eqref{f.sparse} as well as the existence of the cluster point to AIR can be guaranteed by understanding conditions under which the iterates generated by AIR
   is bounded. For this purpose, we need to investigate the asymptotic geometry of $J$ and $X$. 
   In the following a series of results, we discuss the conditions guaranteeing  the  boundedness of $L(J(\xbf^0; \epsbf^0); J_0) $. 
   The concept of horizon cone is a useful tool to characterize the boundedness of a set, which is defined as follows.

   \begin{definition}~\cite[Definition 3.3]{Roc98} Given $Y\subset \mathbb{R}^n$, the horizon cone of $Y$ is 
   	\[ Y^\infty := \{ \zbf \mid \exists t^k\downarrow 0, \{\ybf^k\}\subset Y \text{  such that  }   t^k\ybf^k \to \zbf \}.\]
   \end{definition}
   
   We have the basic properties about horizon cones given in the following proposition, where the first case is trivial  to show and others are from~\cite{Roc98}.  
   
   \begin{proposition}\label{prop.cone}
   	The following hold:
   	\begin{enumerate}
   		\item[(i)] If $X\subset Y \subset \mathbb{R}^n$,  then $X^\infty\subset Y^\infty$. 
   		\item[(ii)]~\cite[Theorem 3.5]{Roc98} The set $Y\subset \mathbb{R}^n$ is bounded if and only if $Y^\infty = \{0\}$. 
   		\item[(iii)]~\cite[Exercise 3.11]{Roc98} Given $Y_i\subset \mathbb{R}^{n_i}$ for $i\in\Gcal$, we have $(Y_1\times \ldots \times Y_m)^\infty 
   		= Y_1^\infty \times \ldots \times Y_m^\infty$.
   		\item[(iv)]~\cite[Theorem 3.6]{Roc98} If $C\subset \mathbb{R}^n$ is non-empty, closed, and convex, then 
   		\[ C^\infty = \{\zbf \mid C+\zbf \subset C\}.\]
   	\end{enumerate}
   \end{proposition}

   Next we investigate the boundedness of $L(J(\xbf^0; \boldsymbol{\epsilon}^0), J_0)$, and provide upper and lower estimates 
   of $L(J(\xbf^0; \boldsymbol{\epsilon}^0), J_0)$.  For this purpose, define 
   \[\begin{aligned}
   H(\xbf^0, \epsbf^0) := &\  \{ \bar \xbf \mid   \bar \xbf \in X^\infty, \bar \xbf \in L(f(\xbf^0); f)^\infty, \\
   &\bar\xbf_i\in L(c_i(\xbf_i^0)+\eps_i^0; c_i)^\infty, i\in\Gcal  \}, 
   and 
   \end{aligned}
   \]
   \[\begin{aligned}
   \tilde H(\xbf^0, \epsbf^0)  := &\ X^\infty \cap L(J(\xbf^0; \epsbf^0); f)^\infty\\
   & \cap (\prod_{i\in\Gcal} L(J(\xbf^0; \epsbf^0)-\underline f; r_i\comp c_i)^\infty ).
   \end{aligned}
   \]
   We now prove the following result about the lower level sets of $L(J(\xbf^0; \boldsymbol{\epsilon}^0), J_0)$. 
   \begin{theorem}\label{thm.level.set}
   	Let $\xbf^0\in X$ and $\epsbf^0\in\mathbb{R}^m_{++}$. Then 
   	\[ L( r_i(c_i(\xbf_i^0)+\eps_i^0); r_i\comp c_i ) = L(c_i(\xbf_i^0)+\eps_i^0; c_i) \]
   	for $i\in\Gcal$. Moreover, it holds that
   	\begin{equation}\label{level.set1}
   	\hat H(\xbf^0, \epsbf^0)   \subset  L(J(\xbf^0; \epsbf^0); J_0)^\infty.
   	\end{equation}
   	Furthermore, suppose $\underline f := \inf_{\xbf\in X} f(\xbf) > -\infty$. Then 
   	\begin{equation}\label{level.set2}
   	L(J(\xbf^0; \epsbf^0); J_0)^\infty \subset \tilde H(\xbf^0, \epsbf^0).
   	\end{equation}
   \end{theorem}
   
   \begin{proof} 
   	The convexity of $L(\xbf_i^0; r_i(c_i(\dotv)+\eps_i^0))$ is by the fact that 
   	\[\begin{aligned}
   	& \xbf_i \in L( r_i(c_i(\xbf_i^0)+\eps_i^0); r_i\comp c_i ) \\
   	\iff  &r_i(c_i(\xbf_i) ) \le r_i(c_i(\xbf_i^0)+\eps_i^0)\\
   	\iff &    c_i(\xbf_i) \le  c_i(\xbf_i^0)+\eps_i^0\\
   	\iff &    \xbf_i \in L(c_i(\xbf_i^0)+\eps_i^0; c_i),
   	\end{aligned}
   	\]
   	where the second equivalence is from the monotonic increasing property of $r_i$. Notice that $L(c_i(\xbf_i^0)+\eps_i^0; c_i)$ is convex. 
   	
   	Now we prove~\eqref{level.set1}.  Let $\xbf\in L(J(\xbf^0; \epsbf^0); J_0)$ and $\bar \xbf$ be an element of  $\hat H(\xbf^0, \epsbf^0)$. 
   	$$ \xbf + \lambda\bar \xbf\in X,\   \xbf + \lambda\bar\xbf \in L(f(\xbf^0); f)^\infty, $$ 
   	and 
   	$$  \xbf_i + \lambda\bar\xbf_i \in L(c_i(\xbf_i^0)+\eps_i^0; c_i)^\infty. $$
   	Therefore, it holds that 
   	\[ \begin{aligned} 
   	J_0( \xbf + \lambda\bar\xbf) = & f(\xbf + \lambda\bar\xbf) + \sum_{i\in\Gcal} r_i(c_i(\xbf_i + \lambda\bar\xbf_i))\\
   	\le & f(\xbf^0) + \sum_{i\in\Gcal} r_i(c_i(\xbf^0_i)+\eps_i^0)\\
   	= & J(\xbf^0; \epsbf^0).
   	\end{aligned}
   	\]
   	Consequently, $\bar\xbf\in L(J(\xbf^0; \epsbf^0); J_0)$, proving \eqref{level.set1}.

   	For \eqref{level.set2}, let $\bar\xbf \in L (J(\xbf^0; \epsbf^0); J_0)^\infty$.   We need to show that   $\bar \xbf$ is an element of  $\tilde H(\xbf^0, \epsbf^0)$.  For this, we may as well assume that $\bar\xbf\neq\mathbf{0}$.  By the fact that $L(J(\xbf^0; \epsbf^0); J_0)^\infty$, there exists $t^k\downarrow 0$ and $\{\xbf^k\}\subset X$ such that $J_0(\xbf^k) \le J(\xbf^0; \epsbf^0)$ and $t^k\xbf^k \to \bar \xbf$. Consequently, $\bar\xbf\in X^\infty$.  Hence 
   	\beq\label{sub.level.set1}
   	L (J(\xbf^0; \epsbf^0); J_0)^\infty \subset X^\infty.
   	\eeq
   	On the other hand, let $\tilde\xbf\in L (J(\xbf^0; \epsbf^0); J_0)$. It then follows that 
   	\[ f(\tilde\xbf) = J_0(\tilde \xbf) - \sum_{i\in\Gcal} r_i(c_i(\tilde\xbf_i)) \le J_0(\tilde \xbf) \le  J(\xbf^0; \epsbf^0), \]
   	where the first inequality is by the fact that $r_i\ge 0$. Consequently, 
   	$\tilde\xbf\in L(J(\xbf^0; \epsbf^0); f)$, implying $L(J(\xbf^0; \epsbf^0); J_0) \subset L(J(\xbf^0; \epsbf^0); f)$. Hence 
   	\beq\label{sub.level.set2}
   	L(J(\xbf^0; \epsbf^0); J_0)^\infty \subset L(J(\xbf^0; \epsbf^0); f)^\infty.
   	\eeq
   	Now consider $c_i$.  We have for $i\in \Gcal$
   	\[ r_i\comp c_i(\tilde \xbf_i) = J_0(\tilde\xbf) - f(\tilde\xbf)-\sum_{j\in\Gcal,j\ne i} r_i(c_i(\tilde\xbf_i)) \le J(\xbf^0; \epsbf^0) - \underline f,\]
   	implying $\tilde\xbf_i \in L(J(\xbf^0; \epsbf^0); r_i\comp c_i)$, $i\in\Gcal$. Therefore, 
   	\[ L(J(\xbf^0; \epsbf^0); J_0) \subset \prod_{i\in\Gcal} L(J(\xbf^0; \epsbf^0) - \underline f; r_i\comp c_i), \]
   	This implies that 
   	$$
   	\begin{aligned}
   	L(J(\xbf^0; \epsbf^0); J_0)^\infty 
   	&\subset  (\prod_{i\in\Gcal} L(J(\xbf^0; \epsbf^0)-\underline f; r_i\comp c_i))^\infty  \\
   	&=   \prod_{i\in\Gcal} L(J(\xbf^0; \epsbf^0)-\underline f; r_i\comp c_i)^\infty,
   	\end{aligned}
   	$$
   	which, combined with \eqref{sub.level.set1} and \eqref{sub.level.set2}, yields \eqref{level.set2}.
   \end{proof}
   
   The following results follow directly from  Theorem~\ref{thm.level.set}. 
   \begin{corollary}\label{coro.bound}
   	If there exists $\bar\xbf\neq0$ such that 
   	\[ \bar \xbf \in X^\infty,\  \bar \xbf \in L(f(\xbf^0); f)^\infty,  \  \bar\xbf_i\in L(c_i(\xbf_i^0)+\eps_i^0; c_i)^\infty, i\in\Gcal, \]
   	then $L(J(\xbf^0;\epsbf^0); J_0)$ is unbounded.   Conversely, if one of the sets 
   	\[ X^\infty, \  L(J(\xbf^0; \epsbf^0); f)^\infty,  \text{ and  }\  (\prod_{i\in\Gcal} L(J(\xbf^0; \epsbf^0)-\underline f; r_i\comp c_i)^\infty )\]
   	is empty, then $L(J(\xbf^0;\epsbf^0); J_0)$ is bounded. 
   \end{corollary}
   
   Based on Corollary~\ref{coro.bound}, 
   we provide    specific  cases in the following proposition  that can guarantee the boundedness of $L(J(\xbf^0;\epsbf^0); J_0)$. 
   
   \begin{proposition}\label{lem.AIR2}  
   	Suppose  
   	$\xbf^0\in  X$ and relaxation vector $\epsbf^0 \in \Rb^m_{++}$. Then the level set  $L(J(\xbf^0, \epsbf^0), J_0)$ is bounded, if  one of the following conditions holds true
   	\begin{enumerate}
   		\item[(i)] $X$ is compact. 
   		\item[(ii)] $f$ is coercive. 
   		\item[(iii)] $r_i\comp c_i$, $i\in\Gcal$ are all coercive. 
   		\item[(iv)] Assume 
   		\[  \gamma_i := \sup_{\|\xbf_i\|\to \infty} r_i(c_i(\xbf_i)) < +\infty, \ i\in\Gcal.\]
   		Suppose $(\xbf^0, \epsbf^0)$ is selected to satisfy $\sum\limits_{i\in\Gcal} r_i(c_i(\xbf_i^0)+\eps_i^0) \le \underline f + \min\limits_i   \gamma_i$.
   	\end{enumerate}
   \end{proposition}
   
   \begin{proof} 
   	Part $(i)$-$(iii)$ are  trivial true by Corollary~\ref{coro.bound}.  We only prove part $(iv)$. 
   	
   	Assume by contradiction that  $L(J(\xbf^0;\epsbf^0); J_0)$ is unbounded, then there exists   $\bar\xbf\in L(J(\xbf^0;\epsbf^0); J_0)^\infty$ with  $\bar\xbf \ne0$.  By the definition of horizon cone, there exists $\{t^k\}\subset \mathbb{R}$ and $\{\xbf^k\}\subset X$ such that 
   	\[ t^k\downarrow 0,  J_0(\xbf^k) \le J(\xbf^0; \epsbf^0),  \text{ and }  t^k\xbf^k\to \bar \xbf.\]
   	Therefore, there must be an $\bar i \in \Gcal$, such that $\|\xbf^k_{\bar i}\|_2 \to \infty$, implying $r_i\comp c_i(\xbf_{\bar i}^k) \to \gamma_{\bar i}$. This means,  
   	$$
   	\begin{aligned}
   	J(\xbf^0;\epsbf^0)   &\ge \lim_{k\to\infty} J_0(\xbf^k) \ge \underline f+ \lim_{k\to\infty} r_i\comp c_i(\xbf_{\bar i}^k) \\
   	&=\underline f+  \gamma_{\bar i} \ge \underline f+ \min\limits_{i\in\Gcal} \gamma_i,
   	\end{aligned}
   	$$
   	a contradiction.  Therefore, $L(J(\xbf^0, \epsbf^0), J_0)$ is bounded. 
   \end{proof}
   
   Proposition~\ref{lem.AIR2}(iv) indicates that the initial iterate $\xbf^0$ and $\epsbf^0$ may need to be chosen sufficiently close to 0 to  enforce convergence 
   if $\phi_i$ is not coercive such as  
   \eqref{f.inv}.  
   
\section{Numerical Experiments}
\label{sec.experiment}

In this section, we test  our proposed  AIR algorithm for nonconvex and nonsmooth sparse optimization problems in two numerical experiments and exhibit its performance. 
In both experiments, the test problems  have   $f(\xbf)\equiv  0$. 
The algorithm is implemented in Matlab with the subproblems solved by the CVX solver~\cite{cvx}.  
We consider  two ways of choosing $r_i$ and $c_i$, $c_i(\xbf_i)=\|\xbf_i\|_1$ and $c_i=\|\xbf_i\|_2^2$,  as described in Table~\ref{tab.1}, so that they can be referred as $\ell_1$-algorithm and $\ell_2$-algorithm, respectively.  In the subproblem, we use identical value for each component of the relaxation parameter $\epsbf^k$, i.e., $\epsbf^k = \epsilon^k \ebf$. In following two experiments, we define \emph{sparsity} as the nonzeros of the vectors.

\subsection{Sparse Signal Recovery}

{

In this subsection, we consider a sparse signal recovery problem~\cite{candes2008enhancing}, which aims to recover sparse vectors from linear measurements. This problem can be formulated as
\begin{equation}\label{f.cs}
\begin{aligned}
\min_{\xbf\in\mathbb{R}^n}\quad & \|\xbf\|_p^p\\
\st \quad & \Abf\xbf = \bbf,
\end{aligned}
\end{equation}
where $\Abf\in \mathbb{R}^{m\times n}$ is the measurement matrix, $\bbf\in \mathbb{R}^{m\times 1}$ is the measurement vector and $p\in(0,1)$.

In the numerical experiments, we fix $n=256$ and the measurement numbers $m=100$. Draw the measurement matrix $\mathbf{A}\in \mathbb{R}^{m\times n}$ with entries normally distributed. Denote $s$ as the number of nonzero entries of $\mathbf{x}_0$, and set $s=40$. We repeat the following procedure $100$ times:
\begin{enumerate}
	\item[(i)] Construct $\mathbf{x}_0\in \mathbb{R}^{n\times 1}$ with randomly zeroing $n-s$ components. Each nonzero entries is chosen randomly from a zero-mean unit-variance Gaussian distribution.
	\item[(ii)] Form $\mathbf{b}_0=\mathbf{A}\mathbf{x}_0$.
	\item[(iii)] Solve the problem for $\hat{\mathbf{x}}_0$.
\end{enumerate}

We compare our AIR algorithms with the iterative reweighted unconstrained $\ell_q$ for sparse vector recovery
(IRucLq-v) algorithm \cite{lai2013improved}. The IRucLq-v algorithm penalizes the linear constraint with a fixed parameter $\lambda$, yielding an unconstrained problem. Then, it uses a reweighted least square method to approximately solve the unconstrained problem, of which the subproblem can be addressed by solving a linear system. We set $\lambda=10^{-6}$, initialize $\epsilon^0=1$ and $\mathbf{x}^0=\mathbf{0}$. Update $\epsilon^{k+1} = \min \{\epsilon^k, \alpha \cdot r(\mathbf{x}^{k+1})_{s+1} \}$, where $r(\mathbf{x}^{k+1})_{s+1}$ denotes the $s+1$ largest (in absolute value) component of $\mathbf{x}^{k+1}$.  Set $s = \lfloor m/2 \rfloor$ and $\alpha=0.9$.

For our AIR algorithms, at each iteration, the subproblem of AIR $\ell_1$-algorithm can be equivalently formulated as a Linear Programming (LP) problem; the subproblem of AIR $\ell_2$-algorithm is a Quadratic Programming (QP) problem. We initialize $\epsilon^0=1$ and $\mathbf{x}^0=\mathbf{0}$. Update $\epsilon^{k+1} = 0.7\epsilon^k$ for AIR $\ell_1$-algorithm, and $\epsilon^{k+1} = 0.9\epsilon^k$ for AIR $\ell_2$-algorithm. We terminate our AIR $\ell_1$-algorithm whenever $ \frac{\|\mathbf{x}^{k+1} - \mathbf{x}^k\|_2}{\|\mathbf{x}^k\|_2} \le 10^{-8}$ and $\epsilon^{k+1}\le 10^{-3}$ are satisfied and record the final objective value as $f(\mathbf{x}_{\ell_1})$. The AIR $\ell_2$-algorithm and the IRucLq-v algorithm are terminated when $\frac{|f(\mathbf{x}^k)-f(\mathbf{x}_{\ell_1})|}{f(\mathbf{x}_{\ell_1})}\le 10^{-3}$ or $k\ge 500$.

We first use one typical realization of the simulation to examine the convergence of all algorithms. We solve the linear system $\mathbf{A}\mathbf{x}= \mathbf{b}$ to get the same feasible initial point for all algorithms. The experimental result is shown in Figure \ref{fig.SSR_obj_iter}. From the result, we observe that the AIR $\ell_1$-algorithm converges faster than the other algorithms, and AIR $\ell_2$-algorithm and IRucLq-v algorithm own similar convergence rates.

\begin{figure}[htb]
	\center
	\includegraphics[scale = 0.28]{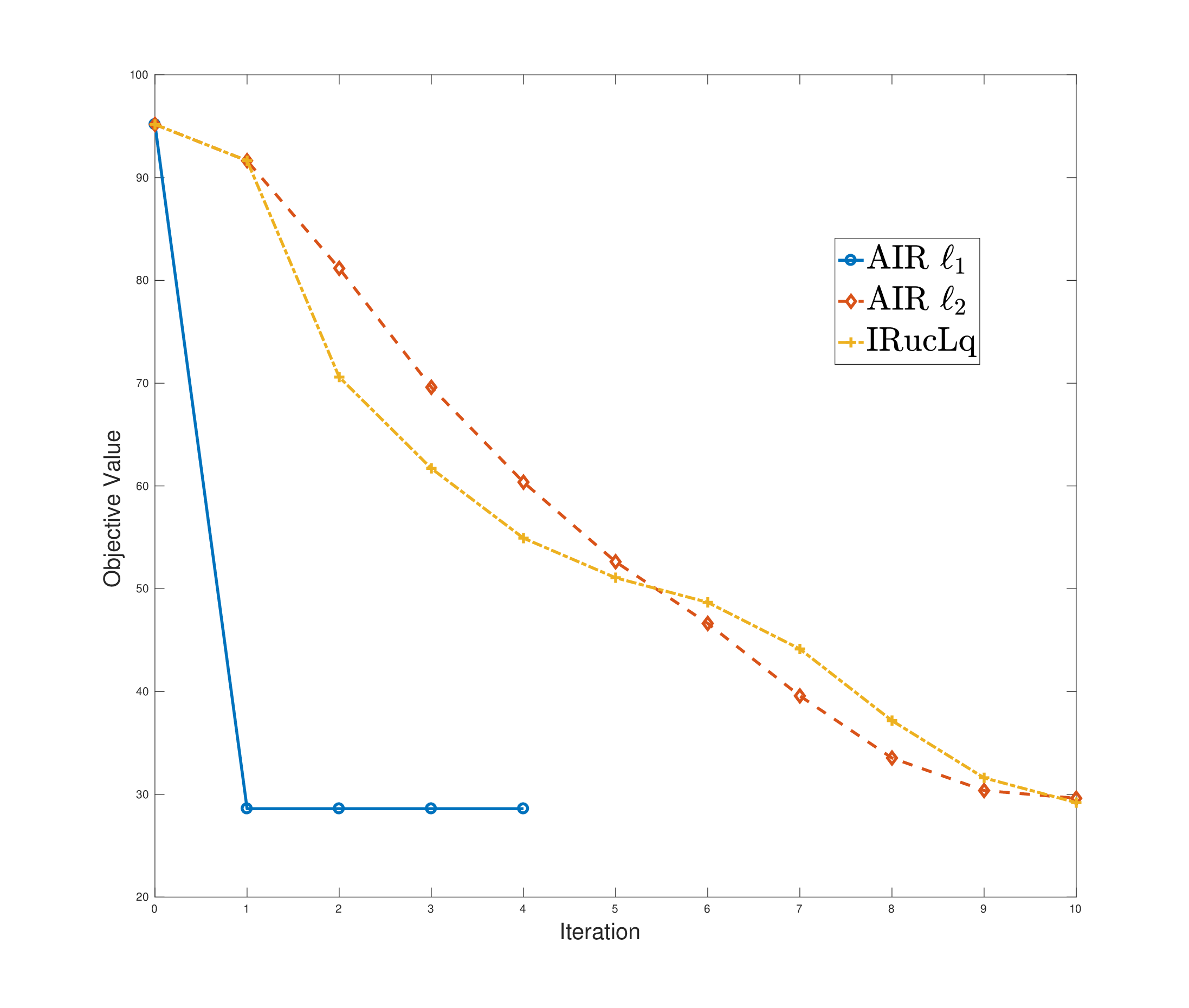}
	\caption{Objective value versus iteration.}
	\label{fig.SSR_obj_iter}
\end{figure}

Then, we further investigate the properties of the solutions generated by all these algorithms. We select the threshold in different levels as
$$\nu \in \{10^{-3},10^{-6},10^{-9},10^{-12}\},$$ and set $x_i=0$ if $|x_i|<=\nu,\ i=1,\cdots,n$ for each $\nu$. The box plots in Figure \ref{fig.SSR_box} demonstrate the statistical properties of the number of nonzeros of all algorithms versus different thresholds.
\begin{figure}[htb]
	\center
	\includegraphics[scale = 0.75]{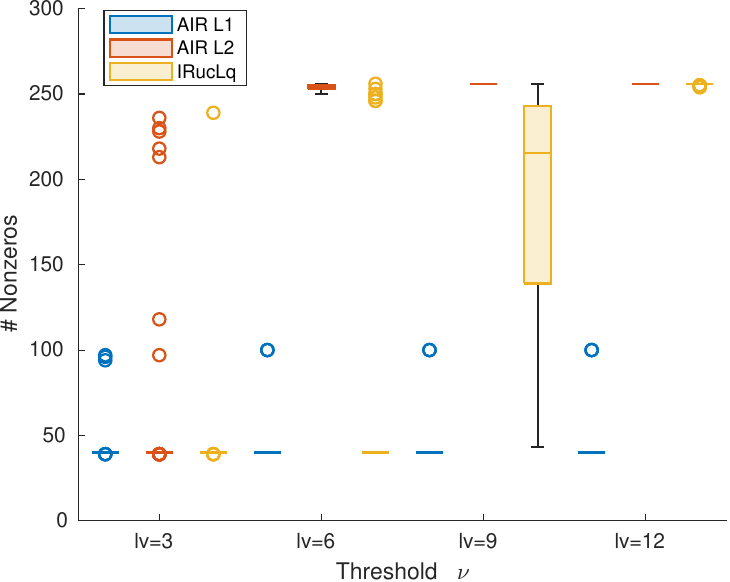}
	\caption{Empirical success probability  versus sparsity.}
	\label{fig.SSR_box}
\end{figure}
The corresponding computation time and constraint violation results are shown in Table \ref{t.SSR_time}. 
\begin{table}[h]
	\caption{
		Average runtime and constraint violation results with associated standard deviation.}
	\label{t.SSR_time}
	\centering
	\begin{tabular}{c|c|c|c}
		\hline
		\multicolumn{1}{l|}{} & \multicolumn{3}{c}{Method}                                                \\ \hline
		                   & AIR $\ell_1$           & AIR $\ell_2$            & IRucLq-v                   \\ \hline
		Time (s)                      & 0.289 $\pm$ 0.035  & 1.672 $\pm$ 0.187  & 0.013 $\pm$ 0.022  \\ \hline
		$\|\mathbf{A}\hat{\mathbf{x}}-\mathbf{b}\|_2$                      & (9.18 $\pm$ 7.31) $\times$ 1e-13  & (5.09 $\pm$ 0.79) $\times$ 1e-14  & (1.76 $\pm$ 0.67) $\times$ 1e-5  \\ \hline
	\end{tabular}
\end{table}

We have the following observations from Fig. \ref{fig.SSR_box} and Table \ref{t.SSR_time}.
\begin{enumerate}
	\item[(i)] From Fig. \ref{fig.SSR_box}, we see that the AIR $\ell_1$-algorithm outputs the most sparse solution. Furthermore, it has the most robust solution with respect to different thresholds. The IRucLq-v algorithm outperforms the AIR $\ell_2$-algorithm both in sparsity and robustness.
	\item[(ii)] From Table \ref{t.SSR_time}, it shows that the IRucLq-v algorithm is more efficient than the AIR algorithms, since it only needs to solve a linear system for each subproblem rather than a LP or QP problem. However, the cost of efficiency is sacrificing the feasibility, which is observed by the constraint violation results.   
	
\end{enumerate}

}

\subsection{Group Sparse Optimization}

In the second experiment, we consider the cloud radio access network power consumption problem~\cite{gsbf}.
In order to solve this problem, a three-stage group sparse beamforming method (GSBF) is proposed in~\cite{gsbf}. The GSBF method solves a group sparse problem in the first stage to induce the group sparsity for the beamformers to guide the remote radio head (RRH) selection. This group sparse problem is addressed by minimizing the mixed $\ell_1/\ell_2$-norm. For further promoting the group sparsity, we replace the $\ell_1/\ell_2$-norm by the $\ell_p/\ell_2$ quasi-norm~\eqref{f.lp} with $p\in(0,1)$~\cite{Shi2016Smoothed}, yielding the following problem
\begin{equation}
\begin{aligned}
\min_{\vbf}\quad & \sum_{l=1}^{L}\sqrt{\frac{\rho_l}{\eta_l}}\|\tilde{\vbf}_l\|_2^p\\
\st \quad & \sqrt{\sum_{i\neq k} \| h_k^{\sf{H}} \vbf_i \|_2^2 + \sigma _k^2}\leq \frac{1}{\gamma _k}\Re (h_k^{\sf{H}} \vbf_k)\\
& \| {\tilde{\vbf}}_{l}  \| _2\leq \sqrt{P_l},\ l =1,\cdots,L, k = 1,\ldots, K.
\end{aligned}\label{f.GS}
\end{equation}
We consider the Cloud-RAN architecture with $L$ remote radio heads (RRHs) and $K$ single-antenna Mobile Users (MUs), where the $l$-th RRH is equipped with $N_l$ antennas. 
$\vbf_{lk} \in \mathbb{C}^{N_l}$ is the transmit beamforming vector from the $l$-th RRH to the $k$-th user with the group structure of transmit vectors ${\tilde{\vbf}}_{l}=[\vbf_{l1}^T,\cdots, \vbf_{lK}^T]^T\in \mathbb{C}^{KN_l\times 1}$. Denote the relative fronthaul link power consumption by $\rho_l$, and the inefficient of drain efficiency of the radio frequency power amplifier by $\eta_l$.
The channel propagation between user $k$ and RRH $l$ is denoted as $\hbf_{lk}\in \mathbb{C}^{N_l}$. 
$P_l$ is the maximum transmit power of the $l$-th RRH.
$\sigma _k$ is the noise at MU $k$.
$\bm{\gamma} = (\gamma_1,..., \gamma_K)^T$ is the target signal-to-interference-plus-noise ratio (SINR).


\subsubsection{Comparison with the mixed $\ell_1/\ell_2$ algorithm}

In this experiment, we compare our AIR $\ell_1$- and $\ell_2$-algorithms with the mixed $\ell_1/\ell_2$ algorithm~\cite{gsbf}.
We consider a network with $L=10$ $2$-antenna RRHs (i.e., $N_l=2$), and $K=6$ single-antenna MUs uniformly and independently distributed in the square region $[-1000,\ 1000]\times[-1000,\ 1000]$ meters. We set $P_l =1$, $\rho_l = 13$, $\eta_l=\frac{1}{4}$ for $l \in\{1,\cdots,L\}$, $\sigma_k = 1$ for $k \in\{1,\ldots, K\}$.
For each quality of service (QoS) $q$ in $\{0,2,4,6\}$, we set the target SINR $\gamma_k = 10^{q/10}$ for $k= 1,\ldots, K$. 
Repeat the following procedure $50$ times:
\begin{enumerate}
	\item[(i)]  Randomly generated network realizations for the channel propagation $\hbf_{lk}$, $l\in\{1,\cdots,L\}$, $k\in\{1,\cdots,K\}$.
	\item[(ii)] Adopt AIR $\ell_1$- and $\ell_2$-algorithm to solve~\eqref{f.GS} for $\tilde{\vbf}^*$.
	\item[(iii)] Regard the $l$-th component of $\tilde{\vbf}^*$ as zero, if $\|\tilde{\vbf}_l^*\|\le 10^{-3}$ for $l\in\{1,\cdots,L\}$.
\end{enumerate}


We set the maximum number of iterations as $T = 500$, $\epsilon_i^0=1$ for AIR and update by $\epsilon^{k+1} = 0.7\epsilon^k$  at each iteration with minimum threshold $10^{-6}$. Set $p=0.1$. The algorithm is terminated whenever
$ |f(\vbf^{k+1}) - f(\vbf^k)| \le 10^{-4}$ is satisfied or $k\ge T$. 

\begin{figure}[hbt]
	\center
	\includegraphics[scale = 0.37]{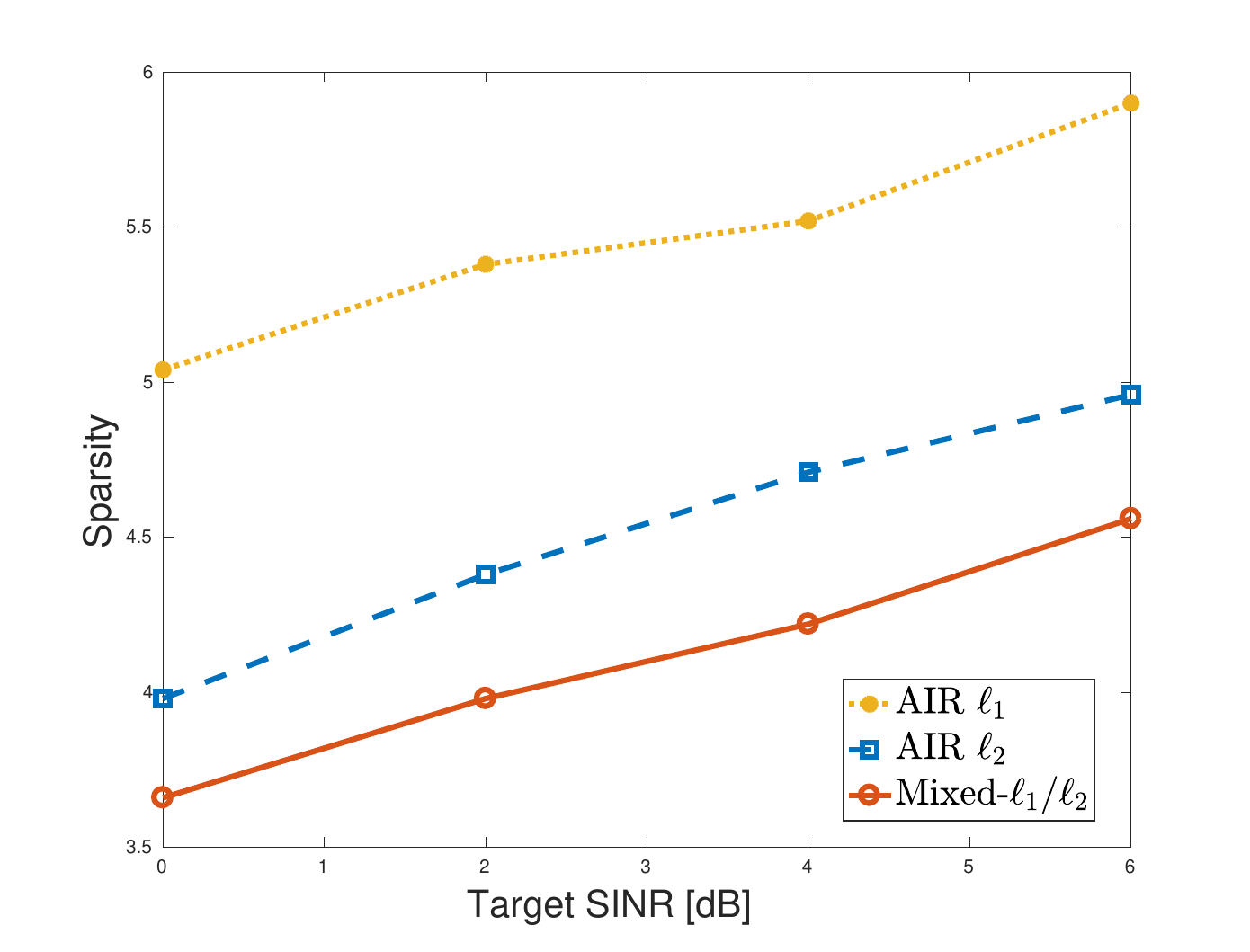}
	\caption{Average sparsity versus target SINR.}
	\label{fig.sparsity}
\end{figure}

In Fig.~\ref{fig.sparsity}, we depict the sparsity of the final solution returned by mixed $\ell_1/\ell_2$ algorithm, AIR $\ell_1$- and $\ell_2$-algorithms for problems with different SINR. The proposed AIR $\ell_1$- and $\ell_2$-algorithms outperform the mixed $\ell_1/\ell_2$ algorithm in promoting the group sparsity. And it is witnessed again that the AIR  $\ell_1$-algorithm outperforms AIR $\ell_2$-algorithm in the ability of  accurately  recovering sparse solutions.

{
\subsubsection{Comparison with the difference of convex algorithm}

We consider the same group sparse optimization problem \eqref{f.GS} in this section, and compare our AIR $\ell_1$- and $\ell_2$-algorithms with the SDCAM~\cite{liu2019successive}. We defer the details of the SDCAM for solving problem \eqref{f.GS} to the appendix. 

In this experiment, we consider a larger-sized network with $L=20$ $2$-antenna RRHs (i.e., $N_l=2$), and $K=15$ single-antenna MUs uniformly and independently distributed in the square region $[-2000,\ 2000]\times[-2000,\ 2000]$ meters. We set $P_l =1$, $\rho_l = 20$, $\eta_l=\frac{1}{4}$ for $l \in\{1,\cdots,L\}$, $\sigma_k = 1$ for $k \in\{1,\ldots, K\}$.
For each quality of service (QoS) $q$ in $\{0,2,4,6\}$, we set the target SINR $\gamma_k = 10^{q/10}$ for $k= 1,\ldots, K$. 
Randomly generated network realizations for the channel propagation $\hbf_{lk}$, $l\in\{1,\cdots,L\}$, $k\in\{1,\cdots,K\}$.

For the AIR algorithms, we set the maximum number of iterations as $T = 500$, $\epsilon_i^0=100$ for AIR and update by $\epsilon^{k+1} = 0.5\epsilon^k$  at each iteration with minimum threshold $10^{-6}$. Set $p=0.1$ (and $p=0.5$). For the case $p=0.1$, the AIR algorithms are terminated whenever both $ \frac{\|\vbf^{k+1} - \vbf^k\|_2}{\|\vbf^k\|_2} \le 10^{-5}$ and $\epsilon^{k+1}\le 10^{-3}$ are satisfied or $k\ge T$. For the case $p=0.5$, the AIR $\ell_1$-algorithm is terminated whenever $ \frac{\|\vbf^{k+1} - \vbf^k\|_2}{\|\vbf^k\|_2} \le 10^{-5}$ and $\epsilon^{k+1}\le 10^{-3}$ are satisfied or $k\ge T$. We denote the solution as $ \vbf_{\ell_1}$ and record the final objective value as $f(\vbf_{\ell_1})$. We terminate the AIR $\ell_2$-algorithm whenever $f(\vbf^{k+1})\le f(\vbf_{\ell_1})$ or $k\ge T$.

The SDCAM applies the Moreau envelope to approximate problem \eqref{f.GS}, which yields a sequence of DC subproblems. They solves the DC subproblems by the Nonmonotone Proximal Gradient method with majorization (NPG$_\text{major}$). 
In SDCAM, we set $\lambda_k = \max\{1/ 10^{k+1}, 10^{-10}\}$ and $\mathbf{v}^\text{feas}$ to be the vector of all ones. In $\text{NPG}_\text{major}$, we set $M=4$, $L_{\max} = 10^8,\ L_{\min} = 10^{-8}$, $\tau = 2$, $c = 10^{-5}$, $L_{k,0}=1$ and for $t\ge 1$,
$$
L_{k,t} = \max\big\{\min\{\frac{{\mathbf{s}^t}^T\mathbf{y}^t}{\|\mathbf{s}^t\|_2^2}, L_{\max}\}, L_{\min} \big\},
$$l
where $\mathbf{s}^t = \mathbf{v}^t - \mathbf{v}^{t-1}$, $\mathbf{y}^t = \nabla h(\mathbf{v}^t) - \nabla h(\mathbf{v}^{t-1})$.
We terminate $\text{NPG}_\text{major}$ when 
$$
\frac{\|\mathbf{s}^t\|_2}{\max(\|\mathbf{v}^t\|_2, 1)}<\epsilon_k\ \text{or}\ \frac{\|F_{\lambda_k}(\mathbf{v}^t) - F_{\lambda_k}(\mathbf{v}^{t-1})\|_2}{\max(|F_{\lambda_k}(\mathbf{v}^t)|, 1)}<10^{-6},
$$
where $\epsilon_0 = 10^{-3}$ and $\epsilon_k = \max\{\epsilon_{k-1}/1.5, 10^{-5}\}$. We terminate the SDCAM whenever $f(\vbf^{k+1})\le f(\vbf_{\ell_1})$ or $k\ge 1000$.

First, we explore the convergence rates of SDCAM and our AIR algorithms with $p=0.5$. We demonstrate the results of objective values versus CPU time for all algorithms in Fig. \ref{fig.obj_time}, where we only report one typical channel realization. The results show that AIR algorithms converge faster than SDCAM, especially the AIR $\ell_1$-algorithm.
\begin{figure}[hbt]
	\center
	\includegraphics[scale = 0.25]{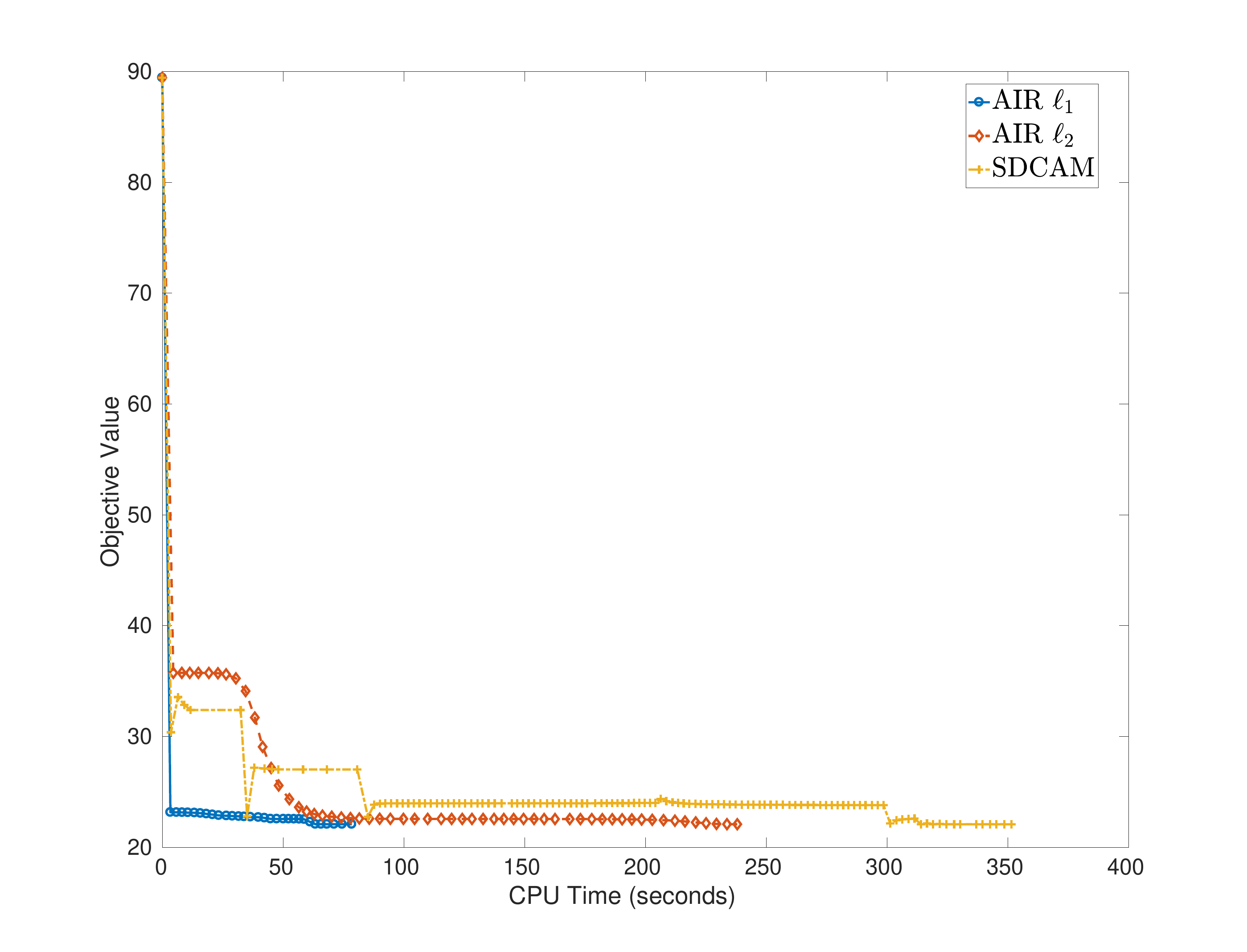}
	\caption{Objective value versus CPU time.}
	\label{fig.obj_time}
\end{figure}

Then, we further investigate the properties of the solutions generated by all these algorithms.  
For each quality of service (QoS) $q$ in $\{0,2,4,6\}$, we randomly generate the network realizations $10$ times. We select the threshold in different levels as
$$\nu \in \{10^{-3}, 10^{-4},10^{-5},10^{-6},10^{-7},10^{-8}\},$$
and set $\|\tilde{\vbf}_l\|_2=0$ if $\|\tilde{\vbf}_l\|_2<=\nu,\ l=1,\cdots,L$. The box plots in Fig. \ref{fig:SDCAM_AIR} demonstrate the distributions of the number of nonzeros of all algorithms versus different Target SINR. The corresponding computation time results are shown in Table \ref{t.time}.

We have the following observations from Fig. \ref{fig:SDCAM_AIR} and Table \ref{t.time}.
\begin{enumerate}
	\item[(i)] From Fig. \ref{fig:p01}, we see that the AIR algorithms with $p=0.1$, especially the AIR $\ell_1$-algorithm, has more sparse solution than the SDCAM. Furthermore, it shows that the solutions generated by our AIR algorithms are more robust with respect to different sparse levels.
	\item[(ii)] From Fig. \ref{fig:p05}, we see that the AIR $\ell_1$-algorithm with $p=0.5$ and the SDCAM perform similarly in this case. The AIR $\ell_1$-algorithm with $p=0.5$ has the best performance both in terms of the sparsity of the solution as well as the robustness of the solution in different levels.
	\item[(iii)] From Table \ref{t.time}, it shows that AIR $\ell_1$-algorithm converges faster than the AIR $\ell_2$-algorithm and the SDCAM. Moreover, the SDCAM fails to achieve the target objective twice for the cases $q=0$, $q=2$ and $q=4$.
\end{enumerate}

\begin{figure}
	\centering
	\begin{subfigure}[b]{0.45\textwidth}
		\centering
		\includegraphics[width=\textwidth]{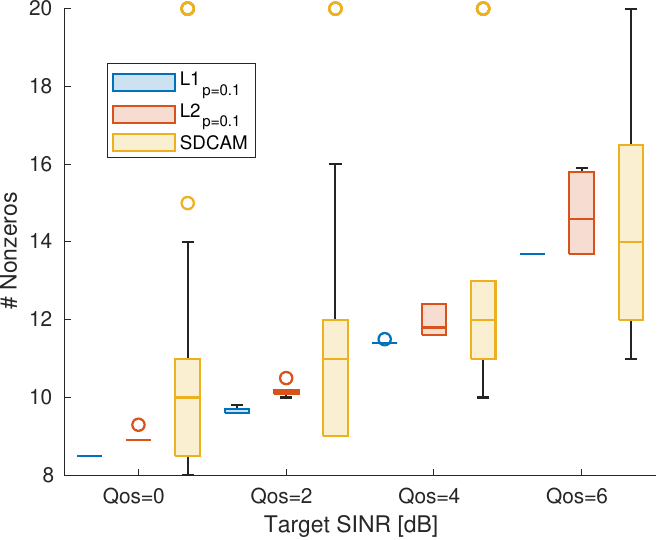}
		\caption{}
		\label{fig:p01}
	\end{subfigure}
	\hfill
	\begin{subfigure}[b]{0.45\textwidth}
		\centering
		\includegraphics[width=\textwidth]{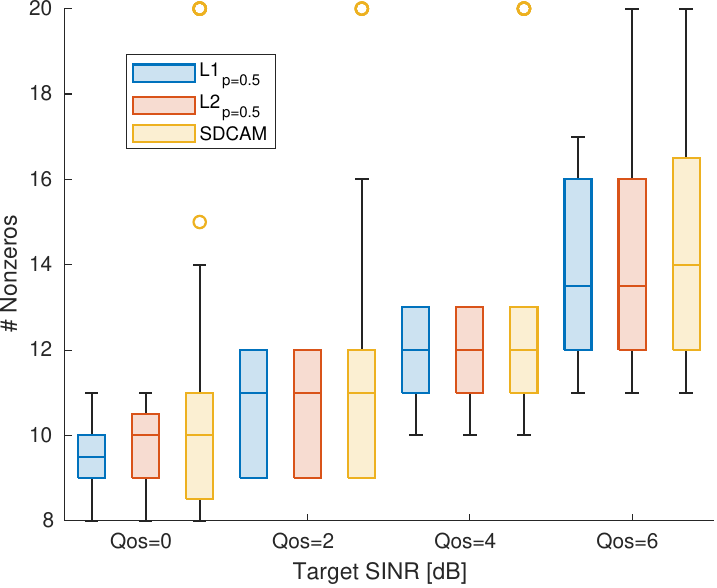}
		\caption{}
		\label{fig:p05}
	\end{subfigure}
	\caption{Comparing SDCAM with AIR algorithms by choosing \eqref{fig:p01} $p=0.1$ and \eqref{fig:p05} $p=0.5$.}
	\label{fig:SDCAM_AIR}
\end{figure}

\begin{table}[ht]
		\caption{Average runtime with associated standard deviation for SDCAM and AIR algorithms (with $p=0.5$). The number in parentheses is the number of the cases that the corresponding algorithms fail to achieve the target objective. Note that for the case SDCAM fails to achieve the target objective, we do not take the corresponding time into the statistical results.}\label{t.time}
		\centering
	\begin{tabular}{c|c|c|c}
		\hline
		\multicolumn{1}{l|}{} & \multicolumn{3}{c}{Method}                                                \\ \hline
		Qos                    & AIR $\ell_1$           & AIR $\ell_2$            & SDCAM                   \\ \hline
		0                      & 106.651$\pm$64.525 (0) & 227.949$\pm$82.590 (0)  & 219.306$\pm$78.987 (2)  \\ \hline
		2                      & 73.070$\pm$51.722 (0)  & 124.946$\pm$41.412 (0)  & 168.794$\pm$70.370 (2)  \\ \hline
		4                      & 74.113$\pm$40.362 (0)  & 206.876$\pm$116.152 (0) & 214.730$\pm$81.465 (2)  \\ \hline
		6                      & 52.764$\pm$11.141 (0)  & 149.554$\pm$70.900 (0)  & 150.455$\pm$109.333 (0) \\ \hline
	\end{tabular}
\end{table}

}

\section{Conclusions}
\label{sec.conclusion}
In this paper, we consider solving  a general formulation for nonconvex and nonsmooth sparse optimization problem, which can take into account  different sparsity-inducing terms.  An iteratively reweighted algorithmic framework  is proposed by solving  a sequence of weighted convex { regularization} subproblems.  We have also derived the optimality condition for the nonconvex and nonsmooth sparse optimization problem  and provided the global convergence analysis for the proposed iteratively reweighted methods.  

Two variants of our proposed algorithm, the $\ell_1$-algorithm and the $\ell_2$-algorithm,  are  implemented and tested. Numerical results exhibits their ability of recovering sparse signals.  It is also witnessed that  the iteratively $\ell_1$-algorithm is generally faster than the $\ell_2$-algorithm because much fewer iterations are needed for $\ell_1$-algorithm. 
Overall, our investigation leads to a variety of interesting research directions:
\begin{itemize}
	\item  A thorough comparison, through either theoretical analysis or numerical experiments, of the existing nonconvex and nonsmooth sparse optimization problems using AIR would be interesting to see.  This should be helpful in providing the guidance for the users to select sparsity-inducing functions.  
	\item Our implementation reduces the relaxation parameter $\epsilon$ by a fraction each time.  It would be useful if a dynamic updating strategy can be derived to reduce the efforts of parameter tuning as well as the sensitivity of the algorithm to $\epsilon$. 
	\item It would be meaningful to have a (local) complexity analysis for the reweigh-ted algorithms.
	
\end{itemize}

{
\begin{acknowledgements}
	Hao Wang's work was supported by National Natural Science Foundation of China [12001367].
	Yaohua Hu's work was supported in part by the National Natural Science Foundation of China [12071306, 11871347], Natural Science Foundation of Guangdong Province of China [2019A1515011917], and Natural Science Foundation of Shenzhen
	[(JCYJ20190808173603590].
\end{acknowledgements}

}

{
\appendix
\section{Implementation Details of SDCAM}

In this section, we provide the details of solving problem \eqref{f.GS} by using SDCAM.
By denoting
$$
\Omega := \{\mathbf{v}: \sqrt{\sum\limits_{i\neq k}\| \mathbf{h}_k^{\sf{H}} \mathbf{v}_i \|_2^2 + \sigma _k^2}\leq \frac{1}{\gamma _k}\Re (\mathbf{h}_k^{\sf{H}} \mathbf{v}_k),\ \| {\tilde{\mathbf{v}}}_{l}  \| _2\leq \sqrt{P_l},\ l =1,\cdots,L,\ k = 1,\cdots, K \},
$$
and
$$
P_1(\mathbf{v}) := \sum_{l=1}^{L}\sqrt{\frac{\rho_l}{\eta_l}}\|\tilde{\mathbf{v}}_l\|_2^p\quad\text{and}\quad P_0(\mathbf{v}) := \delta(\mathbf{v}|\Omega),
$$
we can reformulate problem \eqref{f.GS} as
$$
\min\limits_{\mathbf{v}}\quad F(\mathbf{v}):= P_0(\mathbf{v}) + P_1(\mathbf{v}). 
$$
Then this problem can be solved by the SDCAM \cite{liu2019successive}.
They approximate $F$ by its Moreau envelope at each iteration. More specifically, at iteration $k$, they solve the following approximate problem
$$
\min\limits_{\mathbf{v}}\quad F_{\lambda_k}(\mathbf{v}):=  P_0(\mathbf{v}) + e_{\lambda_k} P_1(\mathbf{v}),
$$
where $e_{\lambda_k} P_1(\mathbf{v})$ is the Moreau envelope of $ P_1(\mathbf{v})$ with parameter $\lambda_k$, which takes the form
$$ e_{\lambda_k} P_1(\mathbf{v}) = := \inf\limits_{\mathbf{x}}\{\frac{1}{2\lambda_k}\|\mathbf{v}-\mathbf{y}\|_2^2 + P_1(\mathbf{y})\}.$$
The SDCAM drives the parameter $\lambda_k$ to $0$ and solves each corresponding subproblem $F_{\lambda_k}$ iteratively. By taking advantage of the equivalently formulation of $e_{\lambda_k} P_1(\mathbf{v})$, i.e.,
$$ e_{\lambda_k} P_1(\mathbf{v}) = \frac{1}{2\lambda}\|\mathbf{v}\|_2^2 - \underbrace{\sup\limits_{\mathbf{y}\in \text{dom}\ P_1}\{\frac{1}{\lambda}\mathbf{v}^T\mathbf{y} - \frac{1}{2\lambda}\|\mathbf{y}\|_2^2-P_1(\mathbf{y})\}}_{D_{\lambda, P_1}(\mathbf{v})},$$
we can reformulate $F_{\lambda_k}$ as a DC problem, which can be solved by the $\text{NPG}_\text{major}$ method. The $\text{NPG}_\text{major}$ method solves the DC subproblem by combing the proximal gradient method with the nonmonotone line search technique, and terminates when the first-order optimality condition is satisfied.
Note that, for the $\text{NPG}_\text{major}$ method, we need to calculate the subgradient of $D_{\lambda, P_1}(\mathbf{v})$, which is proved equal to $\frac{1}{\mu} \text{prox}_{\lambda_k,P_1}(\mathbf{v})$. The proximity operator of $\ell_p$ norm with $p= 1/2$ (or $p=2/3$) has an analytic solution, which is provided in \cite{cao2013fast}.

}

\bibliographystyle{spmpsci}      
\bibliography{ref.bib}   

\begin{thebibliography}{10}
\providecommand{\url}[1]{{#1}}
\providecommand{\urlprefix}{URL }
\expandafter\ifx\csname urlstyle\endcsname\relax
  \providecommand{\doi}[1]{DOI~\discretionary{}{}{}#1}\else
  \providecommand{\doi}{DOI~\discretionary{}{}{}\begingroup
  \urlstyle{rm}\Url}\fi

\bibitem{ahn2017difference}
Ahn, M., Pang, J.S., Xin, J.: Difference-of-convex learning: directional
  stationarity, optimality, and sparsity.
\newblock SIAM Journal on Optimization \textbf{27}(3), 1637--1665 (2017)

\bibitem{ba2014convergence}
Ba, D., Babadi, B., Purdon, P.L., Brown, E.N.: Convergence and stability of
  iteratively re-weighted least squares algorithms.
\newblock IEEE Transactions on Signal Processing \textbf{62}(1), 183--195
  (2014)

\bibitem{bach2012optimization}
Bach, F., Jenatton, R., Mairal, J., Obozinski, G., et~al.: Optimization with
  sparsity-inducing penalties.
\newblock Foundations and Trends{\textregistered} in Machine Learning
  \textbf{4}(1), 1--106 (2012)

\bibitem{bazaraa1974cones}
Bazaraa, M.S., Goode, J., Nashed, M.Z.: On the cones of tangents with
  applications to mathematical programming.
\newblock Journal of Optimization Theory and Applications \textbf{13}(4),
  389--426 (1974)

\bibitem{bian2017optimality}
Bian, W., Chen, X.: Optimality and complexity for constrained optimization
  problems with nonconvex regularization.
\newblock Mathematics of Operations Research \textbf{42}(4), 1063--1084 (2017)

\bibitem{bradley1998feature}
Bradley, P.S., Mangasarian, O.L., Street, W.N.: Feature selection via
  mathematical programming.
\newblock INFORMS Journal on Computing \textbf{10}(2), 209--217 (1998)

\bibitem{EPS}
Burke, J.V., Curtis, F.E., Wang, H., Wang, J.: Iterative reweighted linear
  least squares for exact penalty subproblems on product sets.
\newblock SIAM Journal on Optimization \textbf{25}(1), 261--294 (2015)

\bibitem{candes2008introduction}
Cand{\`e}s, E.J., Wakin, M.B.: An introduction to compressive sampling.
\newblock IEEE signal processing magazine \textbf{25}(2), 21--30 (2008)

\bibitem{candes2008enhancing}
Candes, E.J., Wakin, M.B., Boyd, S.P.: Enhancing sparsity by reweighted
  $\ell_1$ minimization.
\newblock Journal of Fourier analysis and applications \textbf{14}(5-6),
  877--905 (2008)

\bibitem{cao2018unifying}
Cao, S., Huo, X., Pang, J.S.: A unifying framework of high-dimensional sparse
  estimation with difference-of-convex (dc) regularizations.
\newblock arXiv preprint arXiv:1812.07130  (2018)

\bibitem{cao2013fast}
Cao, W., Sun, J., Xu, Z.: Fast image deconvolution using closed-form
  thresholding formulas of lq (q= 12, 23) regularization.
\newblock Journal of visual communication and image representation
  \textbf{24}(1), 31--41 (2013)

\bibitem{chartrand2008iteratively}
Chartrand, R., Yin, W.: Iteratively reweighted algorithms for compressive
  sensing.
\newblock In: Acoustics, speech and signal processing, 2008. ICASSP 2008. IEEE
  international conference on, pp. 3869--3872. IEEE (2008)

\bibitem{chen2010convergence}
Chen, X., Zhou, W.: Convergence of reweighted l1 minimization algorithms and
  unique solution of truncated lp minimization.
\newblock Department of Applied Mathematics, The Hong Kong Polytechnic
  University  (2010)

\bibitem{daubechies2010iteratively}
Daubechies, I., DeVore, R., Fornasier, M., G{\"u}nt{\"u}rk, C.S.: Iteratively
  reweighted least squares minimization for sparse recovery.
\newblock Communications on Pure and Applied Mathematics \textbf{63}(1), 1--38
  (2010)

\bibitem{fan2001variable}
Fan, J., Li, R.: Variable selection via nonconcave penalized likelihood and its
  oracle properties.
\newblock Journal of the American statistical Association \textbf{96}(456),
  1348--1360 (2001)

\bibitem{fazel2003log}
Fazel, M., Hindi, H., Boyd, S.P.: Log-det heuristic for matrix rank
  minimization with applications to hankel and euclidean distance matrices.
\newblock In: American Control Conference, 2003. Proceedings of the 2003,
  vol.~3, pp. 2156--2162. IEEE (2003)

\bibitem{cvx}
Grant, M., Boyd, S., Ye, Y.: Cvx: Matlab software for disciplined convex
  programming (2008).
\newblock (Web page and software.) http://stanford.edu/~boyd/cvx  (2015)

\bibitem{harrell2015ordinal}
Harrell, F.E.: Ordinal logistic regression.
\newblock In: Regression modeling strategies, pp. 311--325. Springer (2015)

\bibitem{hastie2015statistical}
Hastie, T., Tibshirani, R., Wainwright, M.: Statistical learning with sparsity:
  the lasso and generalizations.
\newblock CRC press (2015)

\bibitem{holland1977robust}
Holland, P.W., Welsch, R.E.: Robust regression using iteratively reweighted
  least-squares.
\newblock Communications in Statistics-theory and Methods \textbf{6}(9),
  813--827 (1977)

\bibitem{hu2017group}
Hu, Y., Li, C., Meng, K., Qin, J., Yang, X.: Group sparse optimization via
  $\ell_{p,q}$ regularization.
\newblock The Journal of Machine Learning Research \textbf{18}(1), 960--1011
  (2017)

\bibitem{khajehnejad2011sparse}
Khajehnejad, M.A., Dimakis, A.G., Xu, W., Hassibi, B.: Sparse recovery of
  nonnegative signals with minimal expansion.
\newblock IEEE Transactions on Signal Processing \textbf{59}(1), 196--208
  (2011)

\bibitem{kruger1981varepsilon}
Kruger, A.Y.: $\varepsilon$-semidifferentials and $\varepsilon$-normal
  elements.
\newblock Depon. VINITI \textbf{1331} (1981)

\bibitem{kruger2003frechet}
Kruger, A.Y.: On fr{\'e}chet subdifferentials.
\newblock Journal of Mathematical Sciences \textbf{116}(3), 3325--3358 (2003)

\bibitem{lai2011unconstrained}
Lai, M.J., Wang, J.: An unconstrained $\ell_q$ minimization with $0\leq q \leq
  1$ for sparse solution of underdetermined linear systems.
\newblock SIAM Journal on Optimization \textbf{21}(1), 82--101 (2011)

\bibitem{lai2013improved}
Lai, M.J., Xu, Y., Yin, W.: Improved iteratively reweighted least squares for
  unconstrained smoothed $\backslash$ell\_q minimization.
\newblock SIAM Journal on Numerical Analysis \textbf{51}(2), 927--957 (2013)

\bibitem{image1}
Lanza, A., Morigi, S., Sgallari, F.: Convex image denoising via non-convex
  regularization.
\newblock In: International Conference on Scale Space and Variational Methods
  in Computer Vision, pp. 666--677. Springer (2015)

\bibitem{ling2013decentralized}
Ling, Q., Wen, Z., Yin, W.: Decentralized jointly sparse optimization by
  reweighted $\ell_q$ minimization.
\newblock IEEE Transactions on Signal Processing \textbf{61}(5), 1165--1170
  (2013)

\bibitem{liu2018sparse}
Liu, L., Larsson, E.G., Yu, W., Popovski, P., Stefanovic, C., de~Carvalho, E.:
  Sparse signal processing for grant-free massive connectivity: A future
  paradigm for random access protocols in the internet of things.
\newblock IEEE Signal Processing Magazine \textbf{35}(5), 88--99 (2018)

\bibitem{liu2019successive}
Liu, T., Pong, T.K., Takeda, A.: A successive difference-of-convex
  approximation method for a class of nonconvex nonsmooth optimization
  problems.
\newblock Mathematical Programming \textbf{176}(1), 339--367 (2019)

\bibitem{lobo2007portfolio}
Lobo, M.S., Fazel, M., Boyd, S.: Portfolio optimization with linear and fixed
  transaction costs.
\newblock Annals of Operations Research \textbf{152}(1), 341--365 (2007)

\bibitem{lu2014proximal}
Lu, C., Wei, Y., Lin, Z., Yan, S.: Proximal iteratively reweighted algorithm
  with multiple splitting for nonconvex sparsity optimization.
\newblock In: AAAI, pp. 1251--1257 (2014)

\bibitem{lu2014iterative}
Lu, Z.: Iterative reweighted minimization methods for $\ell_p$ regularized
  unconstrained nonlinear programming.
\newblock Mathematical Programming \textbf{147}(1-2), 277--307 (2014)

\bibitem{lu2017ell}
Lu, Z., Zhang, Y., Lu, J.: $\ell_p $ regularized low-rank approximation via
  iterative reweighted singular value minimization.
\newblock Computational Optimization and Applications \textbf{68}(3), 619--642
  (2017)

\bibitem{malek2014iterative}
Malek-Mohammadi, M., Babaie-Zadeh, M., Skoglund, M.: Iterative concave rank
  approximation for recovering low-rank matrices.
\newblock IEEE Transactions on Signal Processing \textbf{62}(20), 5213--5226
  (2014)

\bibitem{miosso2009compressive}
Miosso, C.J., von Borries, R., Argaez, M., Vel{\'a}zquez, L., Quintero, C.,
  Potes, C.: Compressive sensing reconstruction with prior information by
  iteratively reweighted least-squares.
\newblock IEEE Transactions on Signal Processing \textbf{57}(6), 2424--2431
  (2009)

\bibitem{ochs2015iteratively}
Ochs, P., Dosovitskiy, A., Brox, T., Pock, T.: On iteratively reweighted
  algorithms for nonsmooth nonconvex optimization in computer vision.
\newblock SIAM Journal on Imaging Sciences \textbf{8}(1), 331--372 (2015)

\bibitem{ortega1970iterative}
Ortega, J.M., Rheinboldt, W.C.: Iterative solution of nonlinear equations in
  several variables, vol.~30.
\newblock Siam (1970)

\bibitem{qin2018sparse}
Qin, Z., Fan, J., Liu, Y., Gao, Y., Li, G.Y.: Sparse representation for
  wireless communications: A compressive sensing approach.
\newblock IEEE Signal Processing Magazine \textbf{35}(3), 40--58 (2018)

\bibitem{Roc98}
Rockafellar, R.T., Wets, R.J.B.: Variational analysis, vol. 317.
\newblock Springer Science \& Business Media (2009)

\bibitem{Shi2016Smoothed}
Shi, Y., Cheng, J., Zhang, J., Bai, B., Chen, W., Letaief, K.B.: Smoothed
  $l_p$-minimization for green cloud-ran with user admission control.
\newblock IEEE Journal on Selected Areas in Communications \textbf{34}(4),
  1022--1036 (2016)

\bibitem{shi2018generalized}
Shi, Y., Zhang, J., Chen, W., Letaief, K.B.: Generalized sparse and low-rank
  optimization for ultra-dense networks.
\newblock IEEE Communications Magazine \textbf{56}(6), 42--48 (2018)

\bibitem{gsbf}
Shi, Y., Zhang, J., Letaief, K.B.: Group sparse beamforming for green
  cloud-ran.
\newblock IEEE Transactions on Wireless Communications \textbf{13}(5),
  2809--2823 (2014)

\bibitem{street1988note}
Street, J.O., Carroll, R.J., Ruppert, D.: A note on computing robust regression
  estimates via iteratively reweighted least squares.
\newblock The American Statistician \textbf{42}(2), 152--154 (1988)

\bibitem{wainwright2014structured}
Wainwright, M.J.: Structured regularizers for high-dimensional problems:
  Statistical and computational issues.
\newblock Annual Review of Statistics and Its Application \textbf{1}, 233--253
  (2014)

\bibitem{wang2015optimality}
Wang, H., Li, D.H., Zhang, X.J., Wu, L.: Optimality conditions for the
  constrained $ \ell_p $-regularization.
\newblock Optimization \textbf{64}(10), 2183--2197 (2015)

\bibitem{ml2}
Weston, J., Elisseeff, A., Sch{\"o}lkopf, B., Tipping, M.: Use of the zero-norm
  with linear models and kernel methods.
\newblock Journal of machine learning research \textbf{3}(Mar), 1439--1461
  (2003)

\bibitem{zhang2007penalized}
Zhang, C.H.: Penalized linear unbiased selection.
\newblock Department of Statistics and Bioinformatics, Rutgers University
  \textbf{3} (2007)

\end{thebibliography}

\end{document}